\pdfoutput=1
\RequirePackage{ifpdf}
\ifpdf 
\documentclass[pdftex]{sigma}
\else
\documentclass{sigma}
\fi

\numberwithin{equation}{section}

\newtheorem{Theorem}{Theorem}[section]
\newtheorem{Corollary}[Theorem]{Corollary}
\newtheorem{Lemma}[Theorem]{Lemma}
 { \theoremstyle{definition}
\newtheorem{Definition}[Theorem]{Definition}
\newtheorem{Remark}[Theorem]{Remark} }

\begin{document}

\newcommand{\arXivNumber}{1703.09963}

\renewcommand{\thefootnote}{}

\renewcommand{\PaperNumber}{073}

\FirstPageHeading

\ShortArticleName{Classif\/ication of a Subclass of Two-Dimensional Lattices via Characteristic Lie Rings}

\ArticleName{Classif\/ication of a Subclass of Two-Dimensional\\ Lattices via Characteristic Lie Rings\footnote{This paper is a~contribution to the Special Issue on Symmetries and Integrability of Dif\/ference Equations. The full collection is available at \href{http://www.emis.de/journals/SIGMA/SIDE12.html}{http://www.emis.de/journals/SIGMA/SIDE12.html}}}

\Author{Ismagil HABIBULLIN~$^{\dag\ddag}$ and Mariya POPTSOVA~$^\dag$}

\AuthorNameForHeading{I.~Habibullin and M.~Poptsova}

\Address{$^\dag$~Ufa Institute of Mathematics, 112 Chernyshevsky Str., Ufa 450008, Russia}
\EmailD{\href{mailto:habibullinismagil@gmail.com}{habibullinismagil@gmail.com}}
\EmailD{\href{mailto:mnpoptsova@gmail.com}{mnpoptsova@gmail.com}}

\Address{$^\ddag$~Bashkir State University, 32 Validy Str., Ufa 450076, Russia}

\ArticleDates{Received March 30, 2017, in f\/inal form August 24, 2017; Published online September 07, 2017}

\Abstract{The main goal of the article is testing a new classif\/ication algorithm. To this end we apply it to a relevant problem of describing the integrable cases of a subclass of two-dimensional lattices. By imposing the cut-of\/f conditions $u_{-1}=c_0$ and $u_{N+1}=c_1$ we reduce the lattice $u_{n,xy}=\alpha(u_{n+1},u_n,u_{n-1})u_{n,x}u_{n,y}$ to a f\/inite system of hyperbolic type PDE. Assuming that for each natural $N$ the obtained system is integrable in the sense of Darboux we look for $\alpha$. To detect the Darboux integrability of the hyperbolic type system we use an algebraic criterion of Darboux integrability which claims that the characteristic Lie rings of such a system must be of f\/inite dimension. We prove that up to the point transformations only one lattice in the studied class passes the test. The lattice coincides with the earlier found Ferapontov--Shabat--Yamilov equation. The one-dimensional reduction $x=y$ of this lattice passes also the symmetry integrability test.}


\Classification{37K10; 37K30; 37D99}

\renewcommand{\thefootnote}{\arabic{footnote}}
\setcounter{footnote}{0}

\section{Introduction}

In the present article we study the classif\/ication problem for the following class of two-dimen\-sio\-nal lattices
\begin{gather} \label{eq1}
u_{n,xy}=\alpha(u_{n+1},u_n,u_{n-1} )u_{n,x}u_{n,y}.
\end{gather}
Here the sought function $u=u_{n}(x,y)$ depends on real $x$, $y$ and on integer $n$. Function $\alpha=\alpha(u_{n+1},u_n,u_{n-1} )$ is assumed to be analytical in a domain $D\subset \mathbb{C}^3$. We request also that the derivatives $\frac{\partial\alpha(u_{n+1},u_n,u_{n-1} )}{\partial u_{n+1} }$ and $\frac{\partial\alpha(u_{n+1},u_n,u_{n-1} )}{\partial u_{n-1} }$ do not vanish identically.

Constraint $u_{n_0} =c_0$ where $c_0$ is a constant parameter def\/ines a boundary condition which cuts of\/f the lattice \eqref{eq1} into two independent semi-inf\/inite lattices
\begin{gather}
u_{n,xy} = \alpha(u_{n+1}, u_n, u_{n-1} )u_{n,x} u_{n,y}, \qquad {\rm{for}} \quad n>n_0 \quad (n<n_0),\nonumber\\
u_{n_0} = c_0. \label{eq2}
\end{gather}
Any solutions of the lattice located on the semiaxis $n>n_0$ does not depend on the solutions of that located on $n<n_0$ and vice versa. Turning to the general case of the lattices recall that the boundary conditions (or cut-of\/f constraints) having such a property are called degenerate. It is well known that the degenerate boundary conditions are admitted by any integrable nonlinear lattice. They are compatible with the whole hierarchy of the higher symmetries \cite{AdlerHab, GurelHab}. In the literature they are met in the connection with the so-called open chains (see, for instance,~\cite{Moser}). Since the symmetry approach which is a powerful classif\/ication tool in the dimension $1+1$ (see, for instance, \cite{Adler,Levi,Mikhailov91}) loses its ef\/f\/iciency in higher dimensions (an explanation can be found in~\cite{Mikhailov98}) it became clear years ago that it is necessary to look for alternative classif\/ication algorithms. Since then dif\/ferent approaches to the integrable multidimensional models have been invented (see, for instance, \cite{Bogdanov,Ferapontov2004, Ferapontov2006,Martinez,Sokolov,Pavlov,Pogrebkov,Zakharov}).

In 1994 A.B.~Shabat posed a problem of creating a classif\/ication algorithm by combining the concepts of the degenerate boundary condition, open chain and the characteristic Lie algebra. It is worth mentioning as an important step in this direction the article~\cite{Sh1995} where the structure of the Lie algebra was described for the two-dimensional Toda lattice. Some progress toward creating the classif\/ication method was done in~\cite{H2013}. It was observed that any f\/initely generated subring of the characteristic Lie ring for the integrable case is of f\/inite dimension. The statement was verif\/ied for a large class of the known integrable lattices.

Our interest to the Shabat's problem was stimulated by the success of the method of the hydrodynamic type reductions in the multidimensionality proposed in \cite{Ferapontov2004,Ferapontov2006}. State-of-the-art for the subject and the references can be found in~\cite{Sokolov}.

In the present article the lattice \eqref{eq1} is used as a touchstone for the created algorithm. Our aim is to explain the core of the method and approve its ef\/f\/iciency by solving a relevant classif\/ication problem.

Boundary condition of the form \eqref{eq2} imposed at two dif\/ferent integers $n=N_1$ and $n=N_2$ (take $N_1<N_2-1$) reduces the lattice \eqref{eq1} into a f\/inite system of hyperbolic type equations (open chain)
\begin{gather}
u_{N_1} =c_1, \nonumber\\
u_{n,xy}= \alpha(u_{n+1},u_n,u_{n-1} )u_{n,x}u_{n,y},\qquad N_1 < n < N_2, \label{eq3} \\
u_{N_2} =c_2. \nonumber
\end{gather}

Initiated by the article \cite{H2013}, where a large class of two-dimensional lattices is discussed we use the following
\begin{Definition} \label{definition1}
We call the lattice \eqref{eq1} integrable if the hyperbolic type system \eqref{eq3} obtained from \eqref{eq1} by imposing degenerate boundary conditions is Darboux integrable for any choice of the integers $N_1$, $N_2$.
\end{Definition}
Recall that a system \eqref{eq3} of the hyperbolic type partial dif\/ferential equations is Darboux integrable if it admits the complete set of functionally independent integrals in both of~$x$ and~$y$ directions. Function $I$ of a f\/inite number of the dynamical variables ${\bf{u}},{\bf{u}}_x,{\bf{u}}_y,\ldots$ is a $y$-integral if it satisf\/ies the condition $D_y I = 0$, where $D_y$ is the operator of the total derivative with respect to the variable $y$ and $\bf{u}$ is a vector with the coordinates $u_{N_1+1}, u_{N_1+2},\ldots,u_{N_2-1} $ coinciding with the f\/ield variables. Since the system (\ref{eq3}) is autonomous we can restrict ourselves by considering only autonomous nontrivial integrals. It can be verif\/ied that the $y$-integral does not depend on ${\bf u}_y, {\bf u}_{yy},\ldots$. In what follows we are interested only on nontrivial $y$-integrals, i.e., integrals containing dependence on at least one dynamical variable ${\bf u}, {\bf u}_x, \ldots$. Note that currently the Darboux integrable discrete and continuous models are intensively studied (see, \cite{Yamilov, H2013,H2007,Smirnov2,Smirnov1,Zheltukhin2,Zheltukhin1,ZMHS-UMJ,ZMHSbook,Zhiber2001}).

We justify Def\/inition~\ref{definition1} by the following reasoning. The problem of f\/inding general solution to the Darboux integrable system is reduced to a problem of solving a system of the ordinary dif\/ferential equations. Usually these ODE are explicitly solved. On the other hand side any solution to the considered hyperbolic system \eqref{eq3} is easily prolonged outside the interval $[N_1,N_2]$ and generates a solution of the corresponding lattice \eqref{eq1}. Therefore in this case the lattice~\eqref{eq1} has a large set of the explicit solutions and is def\/initely integrable.

Let us brief\/ly discuss on the content of the article. In Section~\ref{section2} we recall the necessary def\/initions and study the main properties of the characteristic Lie ring which is a basic implement in the theory of the Darboux integrable systems. The goal of Section~\ref{section3} consists in deriving some dif\/ferential equations on the unknown $\alpha$ (it is reasonable to call them integrability conditions) from the f\/inite-dimensionality property of the characteristic Lie ring. To this end we used two test sequences. In Section~\ref{section4} by summarizing the integrability conditions we found the f\/inal form of the searched function $\alpha$. It is remarkable that two test sequences turned out to be enough to complete the classif\/ication. The classif\/ication result is formulated in Theorem~\ref{theorem6} (see Section~\ref{section5}) which claims: any lattice \eqref{eq1} integrable in the sense of Def\/inition~\ref{definition1} can be reduced by an appropriate point transformation $v=p(u)$ to the following one, found earlier in~\cite{Fer-TMF} and, independently, in~\cite{ShY}
\begin{gather}
 v_{n,xy}=v_{n,x}v_{n,y}\left(\frac{1} {v_n-v_{n-1} }-\frac{1} {v_{n+1} -v_{n}}\right). \label{alphalast11}
 \end{gather}
We obtained also a new result concerned to the lattice~(\ref{alphalast11}) by proving that for any choice of the integer $N\geq 0$ the system of the hyperbolic type equations
\begin{gather}
v_{-1} =c_0, \nonumber\\
v_{n,xy}=v_{n,x}v_{n,y}\left(\frac{1} {v_n-v_{n-1} }-\frac{1} {v_{n+1} -v_{n}}\right), \label{eq311} \\
v_{N+1} =c_1, \qquad 0\leq n\leq N \nonumber
\end{gather}
admits a complete set of functionally independent $x$- and $y$-integrals for any constant parame\-ters~$c_0$,~$c_1$, i.e., is Darboux integrable. This fact follows immediately from Theorem~\ref{theorem5} proved in Appendix~\ref{appendixA}, which states that the characteristic Lie rings in both characteristic directions~$x$ and~$y$ for the system (\ref{eq311}) are of f\/inite dimension. In the particular case when $N=1$ for the corresponding system
\begin{gather*}
v_{0,xy}=v_{0,x}v_{0,y}\left(\frac{1} {v_0-c_{0} }-\frac{1} {v_{1} -v_{0}}\right), \qquad
v_{1,xy}=v_{1,x}v_{1,y}\left(\frac{1} {v_1-v_{0} }-\frac{1} {c_{1} -v_{1}}\right) 
\end{gather*}
we give the $y$- and $x$-integrals in an explicit form
\begin{gather*}
I_1 = \frac{v_{0,x}v_{1,x}}{(v_0-c_0)(v_1-v_0)(c_1-v_1)},\qquad
I_2 = \frac{v_{1,xx}}{v_{1,x}}+\frac{v_{0,x}(v_1-c_0)}{(v_0-c_0)(v_1-v_0)} + \frac{2v_{1,x}}{c_1-v_1},\\
J_1 = \frac{v_{0,y}v_{1,y}}{(v_0-c_0)(v_1-v_0)(c_1-v_1)},\qquad
J_2 = \frac{v_{1,yy}}{v_{1,y}}+\frac{v_{0,y}(v_1-c_0)}{(v_0-c_0)(v_1-v_0)} + \frac{2v_{1,y}}{c_1-v_1}.
\end{gather*}

\section{Characteristic Lie rings}\label{section2}

Since the lattice \eqref{eq1} is invariant under the shift of the variable $n$ we can without loss of generality take $N_1 = -1$ and concentrate on the system
\begin{gather}
u_{-1} = c_0, \nonumber\\
u_{n,xy} = \alpha_n u_{n,x} u_{n,y}, \qquad 0 \leq n \leq N, \label{eq2_1} \\
u_{N+1} = c_1. \nonumber
\end{gather}
Here $\alpha_n = \alpha(u_{n-1},u_n,u_{n+1})$. Assume that system \eqref{eq2_1} is Darboux integrable and that\linebreak $I({\bf u},{\bf u}_x,\ldots)$ is its nontrivial integral. Let us evaluate $D_y I$ in the equation $D_y I = 0$ and get due to the chain rule an equation $YI=0$, where
\begin{gather} \label{eq2_2}
Y = \sum_{i=0} ^N \left(u_{i,y} \frac{\partial}{\partial u_i} + f_i \frac{\partial}{\partial u_{i,x}} + f_{i,x}\frac{\partial}{\partial u_{i,xx}} + \cdots \right).
\end{gather}
Here $f_i = \alpha_i u_{i,x} u_{i,y}$. Since the coef\/f\/icients of the equation $YI=0$ depend on $u_{i,y}$ while its solution $I$ does not depend on them we have a system of several linear equations for one unknown~$I$
\begin{gather} \label{eq2_3}
YI=0, \qquad X_j I = 0, \qquad j=1,\ldots,N,
\end{gather}
with $X_i = \frac{\partial}{\partial u_{i,y}}$. It follows from \eqref{eq2_3} that for $\forall\, i$ the operator $Y_i = [X_i, Y ]=X_i Y - Y X_i$ also annihilates $I$. Let us give the explicit form of the operator $Y_i$
\begin{gather*} 
Y_i = \frac{\partial}{\partial u_i} + X_i(f_i) \frac{\partial}{\partial u_{i,x}} + X_i(D_i f_i) \frac{\partial}{\partial u_{i,xx}} + \cdots.
\end{gather*}
Due to the relation $D^k_x f_i = u_{i,y} X_i(D^k_x f_i)$ we represent \eqref{eq2_2} as
\begin{gather}
 Y = \sum_{i=0} ^N u_{i,y} \left(\frac{\partial}{\partial u_i} + X_i(f_i) \frac{\partial}{\partial u_{i,x}} + X_i(D_x f_i) \frac{\partial}{\partial u_{i,xx}} + \cdots \right) = \sum_{i=0} ^N u_{i,y} Y_i. \label{eq2_5}
\end{gather}
The last equation together with \eqref{eq2_3} implies $\sum\limits_{i=0} ^N u_{i,y} Y_i I = 0$. Since the variables $u_{i,y}$ are independent	the coef\/f\/icients of this decomposition all vanish. Now we use the evident relation $[X_k, Y_s ]$ = 0 valid for $\forall\, k,s$. The condition $X_i I= 0$ is satisf\/ied automatically. Thus we arrive at the statement: function $I$ is a $y$-integral of the system \eqref{eq2_1} if and only if it solves the following system of equations
\begin{gather} \label{eq2_6}
Y_i I = 0\qquad {\rm for} \quad i=0,1,\ldots,N.
\end{gather}

Consider the set $R_0(y,N)$ of all multiple commutators of the characteristic vector f\/ields $Y_0, Y_1,\ldots, Y_N$. Denote through $R(y,N)$ the minimal ring containing $R_0(y,N)$. We refer to $R(y,N)$ as the characteristic Lie ring of the system \eqref{eq2_1} in $y$-direction. In a similar way one can def\/ine the characteristic Lie ring in the direction of $x$. Thus we have a complete description of the set of the linear f\/irst order partial dif\/ferential equations the $y$-integral should satisfy to. Now the task is to f\/ind a subset of the linearly independent equations such that all the other equations can be represented as linear combinations of those ones.

We say that the ring $R(y,N)$ is of f\/inite dimension if there exists a f\/inite subset $ \{Z_1, Z_2, \ldots,$ $Z_L \} \subset R(y,N)$ which def\/ines a basis in $R(y,N)$ such that
\begin{enumerate}\itemsep=0pt
\item[1)] every element $Z \in R(y,N)$ is represented in the form $Z = \lambda_1 Z_1 + \dots + \lambda_L Z_L$ with the coef\/f\/icients $\lambda_1, \ldots, \lambda_L$ which might depend on a f\/inite number of the dynamical variables,
\item[2)] relation $\lambda_1 Z_1 +\cdots +\lambda_L Z_L = 0$ implies that $\lambda_1 = \dots = \lambda_L =0$.
\end{enumerate}

Let us formulate now an ef\/fective algebraic criterion	 (see, for instance \cite{ZMHS-UMJ, ZMHSbook}) of solvability of the system \eqref{eq2_6}.

\begin{Theorem} \label{theorem1}
The system \eqref{eq2_1} is Darboux integrable if and only if both characteristic Lie rings $R(x,N)$, $R(y,N)$ are of finite dimension.
\end{Theorem}

\begin{Corollary} \label{corollary1}
The system \eqref{eq2_6} has a nontrivial solution if and only if the ring $R(y,N)$ is of finite dimension.
\end{Corollary}

For the sake of convenience we introduce the following notation ${\rm ad}_X(Z):= [X,Z ]$. We stress that in our further study the operator ${\rm ad}_{D_x}$ plays a crucial role. Below we apply $D_x$ to smooth functions of the dynamical variables ${\bf u}, {\bf u}_x, {\bf u}_{xx}, \ldots$. As it was demonstrated above on this class of functions the operators $D_y$ and $Y$ coincide. Therefore relation $ [D_x, D_y ]=0$ immediately gives $ [ D_x, Y]=0$. Replace now $Y$ due to \eqref{eq2_5} and get
\begin{gather}
[ D_x, Y ] = \sum_{i=0} ^N u_{i,y}( \alpha_i u_{i,x} Y_i + [ D_x, Y_i ] )=0. \label{eq2_7}
\end{gather}
Since in \eqref{eq2_7} the variables $\{u_{i,y} \}_{i=0}^N$ are linearly independent, the coef\/f\/icients should vanish. Consequently we have
\begin{gather}
[ D_x, Y_i ] = -\alpha_i u_{i,x}Y_i. \label{eq2_8}
\end{gather}
From this formula we can easily obtain that ${\rm ad}_{D_x}\colon R(y,N) \rightarrow R(y,N)$. The following lemma describes the kernel of this map (see also~\cite{Sh1995})

\begin{Lemma} \label{lemma1}
 If the vector field
\begin{gather*} 
Z = \sum_i z_{1,i} \frac{\partial}{\partial u_{i,x}} + z_{2,i} \frac{\partial}{\partial	u_{i,xx}} + \cdots
\end{gather*}
satisfies the condition $[ D_x, Z ] = 0$ then $Z=0$.
\end{Lemma}

\section{Method of the test sequences}\label{section3}

We call a sequence of the operators $W_0, W_1, W_2, \ldots$ in $R(y,N)$ a test sequence if the following condition is satisf\/ied for $\forall\,m$
\begin{gather*} 
[D_x, W_m] = \sum_{j=0} ^m w_{j,m}W_j.
\end{gather*}
The test sequence allows one to derive integrability conditions for the hyperbolic type system~\eqref{eq2_1} (see \cite{H2007, ZMHS-UMJ,ZMHSbook}). Indeed, let us assume that \eqref{eq2_1} is Darboux integrable. Then the ring $R(y,N)$ is of f\/inite dimension. Therefore there exists an integer $k$ such that the operators $W_0,\ldots, W_k$ are linearly independent while the operator $W_{k+1} $ is expressed through them as follows
\begin{gather} \label{eq3_2}
W_{k+1} =\lambda_k W_k + \cdots + \lambda_0 W_0.
\end{gather}
Let us apply the operator ${\rm ad}_{D_x}$ to both sides of \eqref{eq3_2}. As a result we f\/ind
\begin{gather*}
 \sum_{j=0} ^k w_{j,k+1} W_j + w_{k+1,k+1} \sum_{j=0} ^k \lambda_j W_j\nonumber\\
 \qquad {} = \sum_{j=0} ^k D_x(\lambda_j)W_j + \lambda_k \sum_{j=0} ^k w_{j,k} W_j + \lambda_{k-1} \sum_{j=0} ^{k-1} w_{j,k-1} W_j+ \cdots+ \lambda_0 w_{0,0} W_0. 
\end{gather*}
By collecting the coef\/f\/icients before the independent operators we obtain a system of the dif\/ferential equations for the coef\/f\/icients $\lambda_0, \lambda_1,\ldots,\lambda_k$. The system is overdetermined since all of the coef\/f\/icients $\lambda_j$ are functions of a f\/inite number of the dynamical variables ${\bf u}, {\bf u}_x, \ldots$. The consistency conditions of this overdetermined system generate integrability conditions for the hyperbolic type system \eqref{eq2_1}. For instance, collecting the coef\/f\/icients before $W_k$ we f\/ind the f\/irst equation of the mentioned system
\begin{gather} \label{eq3_4}
D_x(\lambda_k) = \lambda_k (w_{k+1,k+1} - w_{k,k})+ w_{k,k+1},
\end{gather}
which is also overdetermined.

Below we use two dif\/ferent samples of the test sequences in order to f\/ind the function $\alpha_n$.

\subsection{The f\/irst test sequence}\label{section3.1}

Def\/ine a sequence of the operators in $R(y,N)$ due to the recurrent formula
\begin{gather}
Y_0, \quad Y_1, \quad W_1 = [ Y_0, Y_1 ], \quad W_2 = [Y_0, W_1 ],\quad \ldots,\quad W_{k+1} = [ Y_0, W_k ], \quad \ldots. \label{eq3_5}
\end{gather}
In the case of the f\/irst two members of the sequence we have already deduced commutation relations (see \eqref{eq2_8} above) which are important for our further studies \begin{gather} \label{eq3_6}
 [ D_x, Y_0 ] = -\alpha_0 u_{0,x} Y_0, \qquad [D_x, Y_1 ] = -\alpha_1 u_{1,x} Y_1.
\end{gather}
By using these two relations and applying the Jacobi identity we get immediately
\begin{gather} \label{eq3_7}
[D_x, W_1 ] = -(\alpha_0 u_{0,x} + \alpha_1 u_{1,x})W_1 - Y_0(\alpha_1 u_{1,x})Y_1 + Y_1 (\alpha_0 u_{0,x})Y_0.
\end{gather}
It can be proved by induction that \eqref{eq3_5} is really a test sequence. Moreover it is easily verif\/ied that for $k \geq 2$
\begin{gather*} 
[ D_x, W_k] = p_k W_k + q_k W_{k-1} + \cdots,
\end{gather*}
where the factors $p_k$, $q_k$ are evaluated as follows
\begin{gather*}
p_k = -(\alpha_1 u_{1,x} + k \alpha_0 u_{0,x}), \qquad q_k = \frac{k-k^2} {2} Y_0(\alpha_0 u_{0,x})- Y_0(\alpha_1 u_{1,x})k. 
\end{gather*}
Due to the assumption that $R(y,N)$ is of f\/inite dimension only a f\/inite subset of the sequence~\eqref{eq3_5} is linearly independent. So there exists $M$ such that
\begin{gather} \label{eq3_10}
W_M = \lambda W_{M-1} + \cdots,
\end{gather}
where the operators $Y_0,Y_1,W_1,\ldots, W_{M-1} $ are linearly independent and the tail might contain a~linear combination of the operators $Y_0,Y_1,W_1,\ldots, W_{M-2} $. At the moment we are not interested in that part in~\eqref{eq3_10}.

\begin{Lemma} \label{lemma2}
The operators $Y_0$, $Y_1$, $W_1$ are linearly independent.
\end{Lemma}

\begin{proof}Assume that
\begin{gather*} 
\lambda_1 W_1 + \mu_1 Y_1 + \mu_0 Y_0 = 0.
\end{gather*}
Since the operators $Y_0$, $Y_1$ are of the form $Y_0 = \frac{\partial}{\partial u_0} + \cdots$, $Y_1 = \frac{\partial}{\partial u_1} + \cdots$ while $W_1$ does not contain summands like $\frac{\partial}{\partial u_0} $ and $\frac{\partial}{\partial u_1} $ then the factors $\mu_1$, $\mu_0$ vanish. If in addition $\lambda_1 \neq 0$ then we have $W_1 = 0$. Now by applying the operator ${\rm ad}_{D_x}$ to both sides of this relation we get due to~\eqref{eq3_7} an equation
\begin{gather*}
Y_0(\alpha_1 u_{1,x})Y_1 - Y_1(\alpha_0 u_{0,x})Y_0 = 0,
\end{gather*}
which yields two conditions: $Y_0(\alpha_1 u_{1,x})= \alpha_{1,u_0} u_{1,x} = 0$ and $Y_1(\alpha_0 u_{0,x})=\alpha_{0,u_1} u_{0,x} =0$. Those equalities contradict our assumption that $\frac{\partial \alpha(u_{n+1},u_n,u_{n-1} )}{\partial u_{n \pm 1} } \neq 0$. Lemma is proved. \end{proof}

\begin{Lemma}\label{lemma33}
If the expansion \eqref{eq3_10} holds then
\begin{gather*} 
\alpha(u_1,u_0,u_{-1}) = \frac{P'(u_0)}{P(u_0)+Q(u_{-1} )} + \frac{1} {M-1} \frac{Q'(u_0)}{P(u_1)+Q(u_0)} - c_1(u_0).
\end{gather*}
\end{Lemma}

\begin{proof} It is easy to check that equation \eqref{eq3_4} for the case of the sequence \eqref{eq3_5} takes the following form
\begin{gather}
D_x (\lambda) = -\alpha_0 u_{0,x} \lambda - \frac{M(M-1)}{2} Y_0(\alpha_0 u_{0,x})-M Y_0(\alpha_1 u_{1,x}). \label{eq3_12}
\end{gather}
We simplify the formula \eqref{eq3_12} due to the relations
\begin{gather*}
Y_0(\alpha_0 u_{0,x})= \left(\frac{\partial}{\partial u_0} + \alpha_0 u_{0,x} \frac{\partial}{\partial u_{0,x}} \right)\alpha_0 u_{0,x} = \big(\alpha_{0, u_0} + \alpha_0^2\big)u_{0x},\\ 
Y_0(\alpha_1 u_{1,x})= \alpha_{1,u_0} u_{1,x}. \nonumber
\end{gather*}
A simple analysis of the equation \eqref{eq3_12} gives that $\lambda=\lambda(u_0,u_1)$. Therefore \eqref{eq3_12} gives rise to the equation	 	
\begin{gather*}
\lambda_{u_0} u_{0,x} + \lambda_{u_1} u_{1,x} = -\left( \alpha \lambda + \frac{M(M-1)}{2} \big(\alpha_{0,u_0} +\alpha^2_0\big)\right)u_{0,x} - M \alpha_{1,u_0} u_{1,x}.
\end{gather*}
By comparing the coef\/f\/icients before the independent variables $u_{0,x}$, $u_{1,x}$ we deduce an overdetermined system of the dif\/ferential equations for~$\lambda$
\begin{gather} \label{eq3_14}
\lambda_{u_0} = -\alpha_0 \lambda - \frac{M(M-1)}{2} \big(\alpha_{0, u_0} + \alpha^2_0\big), \qquad \lambda_{u_1} = -M \alpha_{1,u_0}.
\end{gather}
Let us derive and investigate the consistency conditions of the system \eqref{eq3_14}. We dif\/ferentiate the f\/irst equation with respect to $u_{-1} $ and f\/ind
\begin{gather} \label{eq3_15}
\lambda = -\frac{M(M-1)}{2} \frac{\alpha_{0,u_0 u_{-1} }+2 \alpha_0 \alpha_{0, u_{-1} }}{\alpha_{0,u_{-1} }}.
\end{gather}
Since $\lambda_{u_{-1} } = 0$ we have
\begin{gather} \label{eq3_16}
 ( \log \alpha_{0,u_{-1} } )_{u_0 u_{-1} } + 2 \alpha_{0, u_{-1} } = 0.
\end{gather}
Now we introduce a new variable $z$ due to 	the relation $\alpha_{0,u_{-1} } = -\frac{1} {2} e^z$ and reduce \eqref{eq3_16} to the Liouville equation $z_{u_0 u_{-1} } = e^z$ for which we have the general solution	
\begin{gather*}
e^z = \frac{2 P'(u_0)Q'(u_{-1} )}{(P(u_0)+Q(u_{-1} ))^2},
\end{gather*}
where $P(u_0)$ and $Q(u_{-1} )$ are arbitrary dif\/ferentiable functions. Thus for	$\alpha_0$ we can obtain the following explicit expression
\begin{gather} \label{eq3_17}
\alpha_0 = -\frac{1} {2} \int e^z d u_{-1} = \frac{P'(u_0)}{P(u_0) + Q(u_{-1} )} + H(u_0,u_1),
\end{gather}
where $H(u_0,u_1)$ is to be determined. Now we can f\/ind $\lambda$ from the second equation in \eqref{eq3_14}
\begin{gather} \label{eq3_18}
\lambda = - M \int \alpha_{1,u_0} d u_1 = -M \frac{Q'(u_0)}{P(u_1) + Q(u_0)} + M c(u_0).
\end{gather}
Let us specify $H(u_0,u_{1} )$ by replacing in \eqref{eq3_15} $\alpha_0$ and $\lambda$ in virtue of \eqref{eq3_17}, \eqref{eq3_18}. As a result	we obtain
\begin{gather*} 
H(u_0,u_1) = \frac{1} {M-1} \frac{Q'(u_0)}{P(u_1)+Q(u_0)}-\frac{1} {M-1} c(u_0) - \frac{1} {2} \frac{P''(u_0)}{P'(u_0)}.
\end{gather*}

Summarizing the reasonings we can conclude that
\begin{gather} \label{eq3_20}
\alpha(u_1,u_0,u_{-1}) = \frac{P'(u_0)}{P(u_0)+Q(u_{-1} )} + \frac{1} {M-1} \frac{Q'(u_0)}{P(u_1)+Q(u_0)} - c_1(u_0),
\end{gather}
where the functions of one variable $P(u_0)$, $Q(u_0)$, $c_1(u_0) = \frac{1} {M-1} c(u_0) + \frac{1} {2} \frac{P''(u_0)}{P'(u_0)}$ and the integer~$M$ are to be found.
\end{proof}

The next step requires some additional integrability conditions. In what follows we derive them by constructing another test sequence.

\subsection{The second test sequence}\label{section3.2}

Now we concentrate on a test sequence generated by the operators $Y_0$, $Y_1$, $Y_2$ and their multiple commutators. It is more complicated than the previous sequence
\begin{gather}
Z_0 = Y_0, \quad Z_1 = Y_1, \quad Z_2 = Y_2, \quad Z_3 = [Y_1, Y_0 ], \quad Z_4 = [Y_2, Y_1 ], \nonumber\\ Z_5 = [Y_2, Z_3 ], \quad Z_6 = [ Y_1, Z_3 ], \quad Z_7 = [Y_1, Z_4 ], \quad Z_8 = [Y_1, Z_5 ]. \label{eq3_21}
\end{gather}
The members $Z_m$ of the sequence for $m>8$ are def\/ined due to the recurrence $Z_{m} = [Y_1, Z_{m-3} ]$. Note that it is the simplest test sequence generated by the iterations of the map $Z \rightarrow [ Y_1, Z ]$ which contains the operator $ [ Y_2, [ Y_1, Y_0 ] ] = Z_5$.

\begin{Lemma} \label{lemma3}
Operators $Z_0, Z_1, \ldots, Z_5$ constitute a linearly independent set.
\end{Lemma}

\begin{proof}
Firstly we note that the operators $Z_0, Z_1,\ldots, Z_4$ are linearly independent. It can be verif\/ied by using reasonings similar to those from the proof of Lemma~\ref{lemma33}. We prove the lemma by contradiction. Assume that
\begin{gather} \label{eq3_22}
Z_5 = \sum_{j=0} ^4 \lambda_j Z_j.
\end{gather}
Now we specify the action of the operator ${\rm ad}_{D_x}$ on the operators $Z_i$. For $i=0,1,2$ it is obtained from the relation
\begin{gather*}
 [D_x, Y_i ] = -\alpha_i u_{i,x} Y_i.
\end{gather*}
Recall that $\alpha_i = \alpha(u_{i-1},u_i, u_{i+1})$. For $i = 3,4,5$ we have
\begin{gather*}
[ D_x, Z_3 ] = -(a_1 + a_0) Z_3 + \cdots,\\
[D_x, Z_4 ] = -(a_2 + a_1)Z_4 + \cdots,\\
[D_x, Z_5 ] = -(a_0 + a_1 + a_2) Z_5 + Y_0(a_1) Z_4 - Y_2(a_1) Z_3 + \cdots	.
\end{gather*}
Here $a_i = \alpha_i u_{i,x}$. Let us apply the operator ${\rm ad}_{D_x}$ to both sides of \eqref{eq3_22} and obtain
\begin{gather}
 -(a_0 + a_1 + a_2) (\lambda_4 Z_4 + \lambda_3 Z_3 + \cdots) + Y_0(a_1) Z_4 - Y_2(a_1) Z_3 + \cdots \nonumber\\
\qquad{} = \lambda_{4,x} Z_4 + \lambda_{3,x} Z_3 - \lambda_4 (a_1 + a_2) Z_4 -\lambda_3 (a_0 + a_1) Z_3 + \cdots. \label{eq3_23}
\end{gather}
By comparing the coef\/f\/icients before $Z_4$ in \eqref{eq3_23} we obtain the following equation
\begin{gather} \label{eq3_24}
\lambda_{4,x} = -\alpha_0 u_{0,x} \lambda_4 -\alpha_{1,u_0} u_{1,x}.
\end{gather}
A simple analysis of the equation \eqref{eq3_24} shows that $\lambda = \lambda(u_0, u_1)$. Hence the equation \eqref{eq3_24} splits down into two equations $\lambda_{4,u_0} = -\alpha_0 \lambda_4$ and $\lambda_{4,u_1} = -\alpha_{1,u_0} $. The former shows that $\lambda_4 = 0$. Indeed if $\lambda_4 \neq 0$ then we obtain an expression for $\alpha_0$: $\alpha_0 = -\left(\log \lambda_4\right)_{u_0} $ which shows that $(\alpha_0)_{u_{-1} }=0$. It contradicts the assumption that $\alpha(u_1,u_0,u_{-1} )$ depends essentially on~$u_1$ and~$u_{-1} $, therefore $\lambda_4 = 0$. Then~\eqref{eq3_24} implies $\alpha_{1,u_0} = 0$ and it leads again to a~contradiction. \end{proof}

Turn back to the sequence \eqref{eq3_21}. For the further study it is necessary to specify	the action of the operator ${\rm ad}_{D_x}$ on the members of this sequence. It is convenient to divide the sequence into three subsequences and study them separately $\{ Z_{3m} \}$, $\{ Z_{3m+1} \}$, and $\{ Z_{3m+2} \}$.

\begin{Lemma} \label{lemma4}
 Action of the operator ${\rm ad}_{D_x}$ on the sequence \eqref{eq3_21} is given by the following relations	
\begin{gather*}
 [D_x, Z_{3m} ] = - (\alpha_0 u_{0,x} + m \alpha_1 u_{1,x})Z_{3m} \\
\hphantom{[D_x, Z_{3m} ] =}{} +\left( \frac{m-m^2} {2} Y_1 (\alpha_1 u_{1,x})- m Y_1 (\alpha_0 u_{0,x})\right)Z_{3m-3} + \cdots,
\\
 [ D_x, Z_{3m+1} ] = -(\alpha_2 u_{2,x} + m \alpha_1 u_{1,x})Z_{3m+1}\\
\hphantom{[ D_x, Z_{3m+1} ] =}{} + \left( \frac{m-m^2} {2} Y_1 (\alpha_1 u_{1,x})- m Y_1 (\alpha_2 u_{2,x})\right) Z_{3m-2} + \cdots,
\\
 [D_x, Z_{3m+2} ] = - (\alpha_0 u_{0,x} + m \alpha_1 u_{1,x} + \alpha_2 u_{2,x})Z_{3m+2} +Y_0 (\alpha_1 u_{1,x})Z_{3m+1}
+ Y_2 (\alpha_1 u_{1,x})Z_{3m} \\
\hphantom{[D_x, Z_{3m+2} ] =}{}
- (m-1) \left( \frac{m}{2} Y_1 (\alpha_1 u_{1,x})+ Y_1 (\alpha_0 u_{0,x} + \alpha_2 u_{2,x}) \right)Z_{3m-1} + \cdots.
\end{gather*}
\end{Lemma}
Lemma~\ref{lemma4} is easily proved by induction. Since the proof is quite technical we omit it.

\begin{Theorem} \label{theorem2}
Assume that $Z_{3k+2} $ is represented as a linear combination
\begin{gather} \label{eq3_25}
Z_{3k+2} = \lambda_k Z_{3k+1} + \mu_k Z_{3k} + \nu_k Z_{3k-1} + \cdots
\end{gather}
of the previous members of the sequence \eqref{eq3_21} and neither of the operators $Z_{3j+2} $ with $j<k$ is a linear combination of $Z_s$ with $s < 3j+2$. Then the coefficient $\nu_k$ is a solution to the equation
\begin{gather} \label{eq3_26}
D_x(\nu_k) = -\alpha_1 u_{1,x} \nu_k - \frac{k(k-1)}{2} Y_1 (\alpha_1 u_{1,x})- (k-1)Y_1 (\alpha_0 u_{0,x} + \alpha_2 u_{2,x}).
\end{gather}
\end{Theorem}

\begin{Lemma} \label{lemma5}
Suppose that all of the conditions of the theorem are satisfied. In addition assume that the operator $Z_{3k}$ $($operator $Z_{3k+1} )$ is linearly expressed in terms of the operator $Z_i$ with $i<3k$. Then in this decomposition the coefficient before $Z_{3k-1} $ vanishes.
\end{Lemma}

\begin{proof} Assume in contrary that $\lambda \neq 0$ in the formula
\begin{gather} \label{eq3_27}
Z_{3k} = \lambda Z_{3k-1} + \cdots.
\end{gather}
Let us apply ${\rm ad}_{D_x}$ to \eqref{eq3_27}. As a result we f\/ind due to Lemma~\ref{lemma4}
\begin{gather}
 -(\alpha_0 u_{0,x} + k \alpha_1 u_{1,x})\lambda Z_{3k-1} + \cdots\nonumber\\
\qquad{}= D_x(\lambda) Z_{3k-1} - \lambda (\alpha_0 u_{0,x} + (k-1)\alpha_1 u_{1,x} + \alpha_2 u_{2,x})Z_{3k-1} + \cdots.\label{eq3_28}
\end{gather}
Collect the coef\/f\/icients before $Z_{3k-1} $ and obtain an equation the coef\/f\/icient $\lambda$ must satisfy to
\begin{gather*}
D_x(\lambda) = \lambda (\alpha_2 u_{2,x} - \alpha_1 u_{1,x}).
\end{gather*}
Due to our assumption above $\lambda$ does not vanish and hence
\begin{gather} \label{eq3_29}
D_x(\log \lambda) = \alpha_2 u_{2,x} - \alpha_1 u_{1,x}.
\end{gather}
Since $\lambda$ depends on a f\/inite number of the dynamical variables then due to equation \eqref{eq3_29} $\lambda$~might depend only on $u_1$ and $u_2$. Therefore \eqref{eq3_28} yields
\begin{gather*}
(\log \lambda)_{u_1} u_{1,x} + (\log \lambda)_{u_2} u_{2,x} = \alpha_2 u_{2,x} - \alpha_1 u_{1,x}.
\end{gather*}
The variables $u_{1,x}$, $u_{2,x}$ are independent, so the last equation implies $\alpha_1 = -(\log \lambda)_{u_1} $, $\alpha_2 = (\log \lambda)_{u_2} $. Thus $\alpha_1 = \alpha_1(u_1, u_2)$ depends only on $u_1$, $u_2$. It contradicts our assumption that $\alpha_1$ depends essentially on $u_0$. The contradiction shows that assumption $\lambda \neq 0$ is not true. That completes the proof. \end{proof}

Now in order to prove Theorem~\ref{theorem2} we apply the operator ${\rm ad}_{D_x}$ to both sides of \eqref{eq3_25} and then simplify due to the relation from Lemma~\ref{lemma4}. Comparison of the coef\/f\/icients before~$Z_{3k-1} $ implies equation \eqref{eq3_26}.

Let us f\/ind the explicit expressions for the coef\/f\/icients of the equation \eqref{eq3_26}
\begin{gather*}
Y_1 (\alpha_0 u_{0,x})= \alpha_{0, u_1} u_{0,x}, \qquad Y_1 (\alpha_2 u_{2,x})= \alpha_{2,u_1} u_{2,x}, \qquad Y_1(\alpha_1 u_{1,x})= \big(\alpha_{1, u_1} + \alpha^2_1\big) u_{1,x}
\end{gather*}
and substitute them into \eqref{eq3_26}
\begin{gather}
 D_x(\nu_k) = -\alpha_1 \nu_k u_{1,x} - \frac{k(k-1)}{2} \big(\alpha_{1,u_{1} }+ \alpha_1^2\big) u_{1,x} - (k-1)(\alpha_{0,u_1} u_x + \alpha_{2,u_1} u_{2,x}). \label{eq3_30}
\end{gather}
A simple analysis of \eqref{eq3_30} convinces that $\nu_k$ might depend only on the variables $u_0$, $u_1$, $u_2$. Therefore	
\begin{gather} \label{eq3_31}
D_x(\nu_k) = \nu_{k,u_0} u_{0,x} + \nu_{k,u_1} u_{1,x} + \nu_{k, u_2} u_{2,x}.
\end{gather}
From the equations \eqref{eq3_30}, \eqref{eq3_31} we obtain a system of the equations for the coef\/f\/icient $\nu_k$
\begin{gather}
\nu_{k,u_0} = -(k-1)\alpha_{0,u_1}, \label{eq3_32} \\
\nu_{k,u_1} = -\alpha_1 \nu_k - \frac{k(k-1)}{2} \big(\alpha_{1,u_1} + \alpha^2_1\big), \label{eq3_33} \\
\nu_{k,u_2} = -(k-1)\alpha_{2,u_1}. \label{eq3_34}
\end{gather}
Substitute the preliminary expression for the function $\alpha$ given by the formula \eqref{eq3_20} into the equation \eqref{eq3_32} and get
\begin{gather*}
\nu_{k,u_0} = \frac{k-1} {M-1} \frac{P'(u_1) Q'(u_0)}{(P(u_1)+ Q(u_0))^2}.
\end{gather*}
Integration of the latter with respect to $u_0$ yields
\begin{gather*} 
\nu_k = -\frac{k-1} {M-1} \frac{P'(u_1)}{P(u_1) + Q(u_0)} + H(u_1,u_2).
\end{gather*}
Since $\nu_{k, u_2} =H_{u_2} $ the equation \eqref{eq3_34} gives rise to the relation
\begin{gather*}
H_{u_2} =(k-1)\frac{P'(u_2)Q'(u_1)} {(P(u_2)+Q(u_1))^2}.
\end{gather*}
Now by integration we obtain an explicit formula for $H$
\begin{gather*}
H = -(k-1) \left( \frac{Q'(u_1)}{P(u_2)+Q(u_1)} + A(u_1) \right),
\end{gather*}
which produces
\begin{gather*} 
\nu_k = -(k-1) \left(\frac{1} {M-1} \frac{P'(u_1)}{P(u_1) + Q(u_0)} + \frac{Q'(u_1)}{P(u_2) + Q(u_1)} + A(u_1) \right).
\end{gather*}
Let us substitute the values of $\alpha$ and $\nu_k$ found into the equation \eqref{eq3_33}. We get a huge equation
\begin{gather}
 -\frac{(k-1)}{M-1} \left( \frac{P''(u_1)}{P(u_1)+Q(u_0)} - \frac{P'^2(u_1)}{(P(u_1)+Q(u_0))^2} \right)\nonumber\\
\qquad\quad{}-(k-1)\left(\frac{Q''(u_1)}{P(u_2)+Q(u_1)} - \frac{Q'^2(u_1)}{(P(u_2)+Q(u_1))^2} + A'(u_1) \right) \nonumber\\
\qquad{} =(k-1)\left(\frac{P'(u_1)}{P(u_1)+Q(u_0)} + \frac{1} {M-1} \frac{Q'(u_1)}{P(u_2)+Q(u_1)} -c_1(u_1) \right) \nonumber\\
\qquad\quad{}\times \left( \frac{1} {M-1} \frac{P'(u_1)}{P(u_1)+Q(u_0)}+ \frac{Q'(u_1)}{P(u_2)+Q(u_1)} + A(u_1)\right) \nonumber\\
\qquad\quad{}- \frac{k(k-1)}{2} \left( \frac{P''(u_1)}{P(u_1)+Q(u_0)} + \frac{1} {M-1} \frac{Q''(u_1)}{P(u_2)+Q(u_1)} \right.\nonumber\\
\qquad\quad{} - \frac{1} {M-1} \frac{Q'^2(u_1)}{(P(u_2)+Q(u_1))^2} + \frac{1} {M-1} \frac{2 Q'(u_1) P'(u_1)}{(P(u_1)+Q(u_0))(P(u_2)+Q(u_1))} \nonumber\\
 \qquad\quad{} + \frac{1} {(M-1)^2} \frac{Q'^2(u_1)}{(P(u_2)+Q(u_1))^2} \nonumber\\
\left.\qquad\quad{} -c'_1(u_1) - 2 c_1(u_1)\left( \frac{P'(u_1)}{P(u_1)+Q(u_0)} + \frac{1} {M-1} \frac{Q'(u_1)}{P(u_2)+Q(u_1)} \right) + c^2_1(u_1) \right). \label{eq3_37}
\end{gather}
Evidently due to our assumption $\frac{\partial}{\partial u_1} \alpha(u_1,u_0,u_{-1}) \neq 0$, $\frac{\partial}{\partial u_{-1} } \alpha(u_1,u_0,u_{-1}) \neq 0$ the functions~$P'(u_2)$ and~$Q'(u_0)$ do not vanish. Therefore the variables
\begin{gather*}
\frac{Q'^2(u_1)}{(P(u_2)+Q(u_1))^2}, \qquad \frac{P'^2(u_1)}{(P(u_1)+Q(u_0))^2}, \qquad \frac{P'(u_1)Q'(u_1)}{(P(u_1)+Q(u_0))(P(u_2)+Q(u_1))}
\end{gather*}
are independent. By gathering the coef\/f\/icients before these variables in \eqref{eq3_37} we get a system of two equations
\begin{gather} \label{eq3_38}
\left( 1 - \frac{1} {M-1} \right) \left(1 - \frac{k}{2(M-1)} \right) = 0, \qquad 1+ \frac{1} {(M-1)^2} = \frac{k}{M-1}.
\end{gather}
There are two solutions to the system \eqref{eq3_38}: $M=0$, $k=-2$ and $M=2$, $k=2$. The former does not f\/it since~$k$ should be positive, so we have the only possibility $M=2$, $k=2$. This f\/inishes the proof of Theorem~\ref{theorem2}.

\section[Finding the functions $P$, $Q$ and $c_1$]{Finding the functions $\boldsymbol{P}$, $\boldsymbol{Q}$ and $\boldsymbol{c_1}$}\label{section4}

In this section we specify the function $\alpha$ given by \eqref{eq3_20}. For this aim we should consider expansions \eqref{eq3_10}, \eqref{eq3_25} using the fact that $M=2$, $k=2$.

Let us rewrite the expansion \eqref{eq3_10} in the complete form
\begin{gather} \label{eq4_2}
W_2 = \lambda W_1 + \sigma Y_1 + \delta Y_0.
\end{gather}

\begin{Theorem} \label{theorem3}
Expansion \eqref{eq4_2} holds if and only if the function $\alpha$ in \eqref{eq1} is of the following form
\begin{gather} \label{Theor-I20}
 \alpha(u_{n+1},u_n,u_{n-1}) = \frac{P'(u_n)}{P(u_n)+Q(u_{n-1} )} + \frac{Q'(u_n)}{P(u_{n+1} )+Q(u_n)}
-\frac{1} {2} \left( \log Q'(u_n)P'(u_n) \right)',\!\!\!
\end{gather}
where the functions $P(u_n)$, $Q(u_n)$ are connected with each other by the differential constraint
\begin{gather} \label{Theor-I21}
-3Q''^2P'^2 - 2P'''P'Q'^2 + 3P''^2Q'^2 + 2 P'^2Q'''Q' = 0.
\end{gather}
\end{Theorem}

\begin{proof}
 Firstly by using relations \eqref{eq3_6}, \eqref{eq3_7} and applying the Jacobi identity we get
\begin{gather*}
[ D_x, W_2]= -(2 a_0 + a_1)W_2 - Y_0(a_0 + 2a_1)W_1 + (2 Y_0 Y_1 (a_0) - Y_1 Y_0 (a_0)) Y_0 - Y_0 Y_0 (a_1)Y_1.
\end{gather*}

Evidently only one summand in \eqref{eq4_2} contains the term $\frac{\partial}{\partial u_1} $, namely $\sigma Y_1$, and only one summand contains the term $\frac{\partial}{\partial u_0}$, namely $\delta Y_0$. Hence $\sigma=0$, $\delta=0$ and we have
\begin{gather*}
W_2 = \lambda W_1.
\end{gather*}
Now by applying the operator ${\rm ad}_{D_x}$ to both sides of this relation we obtain
\begin{gather*}
 -(2 a_0 + a_1)W_2 - Y_0(a_0 + 2a_1)W_1 + (2 Y_0 Y_1 (a_0) - Y_1 Y_0 (a_0)) Y_0 - Y_0 Y_0 (a_1)Y_1 \\
\qquad{}=D_x(\lambda) W_1 + \lambda ( -(a_0 + a_1) W_1 + Y_1(a_0) Y_0 - Y_0(a_1)Y_1 ).
\end{gather*}
Collecting the coef\/f\/icients before $W_2$, $W_1$, $Y_1$, and $Y_0$ we f\/ind the following system
\begin{gather}
D_x(\lambda) = -a_0 \lambda - Y_0 (a_0 + 2 a_1), \label{eq4_3} \\
-Y_0 Y_0 (a_1) = -\lambda Y_0 (a_1), \label{eq4_4} \\
2 Y_0 Y_1 (a_0) - Y_1 Y_0 (a_0) = \lambda Y_1(a_0). \label{eq4_5}
\end{gather}
Setting $M=2$ in \eqref{eq3_12} we obtain equation \eqref{eq4_3}. The overdetermined system \eqref{eq3_14} takes the form
\begin{gather}
\lambda_{u_0} = -\alpha_0 \lambda - \big(\alpha_{0,u_0} + \alpha^2_0\big),\label{eq4_5_0} \\
\lambda_{u_1} = - 2 \alpha_{1,u_0}.\nonumber 
\end{gather}
 Thus
\begin{gather}
\lambda = -2 \frac{Q'(u_0)}{P(u_1) + Q(u_0)} + 2 c(u_0),\label{eq4_6} \\
\alpha(u_1,u_0,u_{-1}) = \frac{P'(u_0)}{P(u_0)+Q(u_{-1} )} + \frac{Q'(u_0)}{P(u_1)+Q(u_0)} - \frac{1} {2} \frac{P''(u_0)}{P'(u_0)} - c(u_0).\label{eq4_7}
\end{gather}
We rewrite \eqref{eq4_4}, \eqref{eq4_5} due to the relations
\begin{gather*}
Y_0(a_0) = \left( \frac{\partial}{\partial u_0} + \alpha_0 u_{0,x} \frac{\partial}{\partial u_{0,x}} + \cdots \right)(\alpha_0 u_{0,x}) = \big(\alpha_{0,u_0} + \alpha^2_0 \big)u_{0,x},\\
Y_0(a_1) = \left( \frac{\partial}{\partial u_0} + \alpha_0 u_{0,x} \frac{\partial}{\partial u_{0,x}} + \cdots \right)(\alpha_1 u_{1,x})= \alpha_{1,u_0} u_{1,x},\\
Y_0(a_0+2a_1) = Y_0(a_0) + 2 Y_0(a_1) = \big(\alpha_{0,u_0} + \alpha^2_0\big)u_{0,x} + 2 \alpha_{1,u_0} u_{1,x},\\
Y_0 Y_0 (a_1) = \left( \frac{\partial}{\partial u_0} + \alpha_0 u_{0,x} \frac{\partial}{\partial u_{0,x}} + \cdots \right)(\alpha_{1,u_0} u_{1,x})= \alpha_{1,u_0u_0} u_{1,x},\\
Y_1(a_0) = \left( \frac{\partial}{\partial u_1} + \alpha_1 u_{1,x} \frac{\partial}{\partial u_{1,x}} + \cdots \right) (\alpha_0 u_{0,x}) = \alpha_{0,u_1} u_{0,x},\\
Y_0 Y_1(a_0) = \left( \frac{\partial}{\partial u_0} + \alpha_0 u_{0,x} \frac{\partial}{\partial u_{0,x}} + \cdots \right) (\alpha_{0,u_1} u_{0,x}) = \big( \alpha_{0,u_0u_1} + \alpha_0 \alpha_{0,u_1} \big) u_x,\\
Y_1 Y_0 (a_0) = \left( \frac{\partial}{\partial u_1} + \alpha_1 u_{1,x} \frac{\partial}{\partial u_{1,x}} + \cdots \right) \big( \big(\alpha_{0,u_0} + \alpha^2_0 \big)u_{0,x} \big)
 = \big( \alpha_{0,u_0u_1} + 2 \alpha_0 \alpha_{0,u_1} \big)u_{0,x}
\end{gather*}
as follows
\begin{gather} \label{eq4_8}
\alpha_{1,u_0 u_0} =\lambda \alpha_{1,u_0}, \qquad \alpha_{0,u_0u_1} = \lambda \alpha_{0,u_1}.
\end{gather}
We substitute \eqref{eq4_6}, \eqref{eq4_7} into \eqref{eq4_8} and f\/ind that $c(u_0) = \frac{1} {2} \frac{Q''(u_0)}{Q'(u_0)}$. So we f\/ind that functions \eqref{eq4_6}, \eqref{eq4_7} are given by
\begin{gather}
 \lambda_{(R)} = \lambda(u_0,u_1) = -2 \frac{Q'(u_0)}{P(u_1) + Q(u_0)} + \frac{Q''(u_0)}{Q'(u_0)},\label{eq4_9}\\
 \alpha(u_1,u_0,u_{-1}) = \frac{P'(u_0)}{P(u_0)+Q(u_{-1} )} + \frac{Q'(u_0)}{P(u_1)+Q(u_0)} - \frac{1} {2} \frac{P''(u_0)}{P'(u_0)} - \frac{1} {2} \frac{Q''(u_0)}{Q'(u_0)}.\label{eq4_10}
\end{gather}
Substituting \eqref{eq4_9}, \eqref{eq4_10} into \eqref{eq4_5_0} we obtain that
the functions $P$, $Q$ must satisfy the equality
\begin{gather*}
-3Q''^2P'^2 - 2P'''P'Q'^2 + 3P''^2Q'^2 + 2 P'^2Q'''Q' = 0.
\end{gather*}

Thus we have proved that if the expansion \eqref{eq3_10} holds then it should be of the form
\begin{gather*} 
W_2 = \lambda_{(R)} W_1.
\end{gather*}
Or the same
\begin{gather*} 
[Y_0, [Y_1,Y_0] ] = \lambda_{(R)} [Y_1,Y_0].\tag*{\qed}
\end{gather*}\renewcommand{\qed}{}
\end{proof}

Let us def\/ine a sequence of the operators in $R(y,N)$ due to the following recurrent formula
\begin{gather*}
Y_0, \quad Y_1, \quad \tilde{W}_1 = [Y_1, Y_0], \quad \tilde{W}_2 = [Y_1, W_1],\quad \ldots,\quad \tilde{W}_{k+1} = [Y_1, \tilde{W}_k], \quad\ldots.
\end{gather*}
It slightly dif\/fers from \eqref{eq3_5} and can be studied in a similar way. We can easily check that the conditions \eqref{Theor-I20}, \eqref{Theor-I21} provide the representation
\begin{gather*}
\tilde{W}_2 = \lambda_{(L)} \tilde{W}_1.
\end{gather*}
Or the same
\begin{gather} \label{H-I_Y1Y10}
[Y_1, [Y_1, Y_0] ] = \lambda_{(L)} [Y_1, Y_0]
\end{gather}
with the coef\/f\/icient
\begin{gather*}
\lambda_{(L)} = -\frac{2 P'(u_1)}{P(u_1) + Q(u_0)} + \frac{P''(u_1)}{P'(u_1)}.
\end{gather*}

Let us consider expansion \eqref{eq3_25} setting $k=2$,
\begin{gather} \label{eq4_13}
Z_8 = \lambda Z_7 + \mu Z_6 + \nu Z_5 + \rho Z_4 + \kappa Z_3 + \sigma Z_2 + \delta Z_1 + \eta Z_0.
\end{gather}

\begin{Theorem} \label{theorem4}
 Expansions \eqref{eq4_2}, \eqref{eq4_13} hold if and only if the function $\alpha$ in \eqref{eq1} is of one of the forms
\begin{gather}
\alpha_0=\alpha(u_1,u_0,u_{-1}) = \frac{P'(u_0)}{P(u_0) + c_1 P(u_{-1}) + c_2} + \frac{c_1 P'(u_0)}{P(u_1) + c_1 P(u_0) + c_2} - \frac{P''(u_0)}{P'(u_0)}, \label{alpha1}\\
\alpha_0=\alpha(u_1,u_0,u_{-1}) = \frac{c_3 r(u_{-1}) r'(u_0)}{c_3 r(u_0) r(u_{-1}) + c_4 r(u_{-1}) - c_1 + c_2 r(u_{-1} )} \nonumber\\
\hphantom{\alpha_0=}{} + \frac{c_1 r'(u_0)}{r(u_0) \bigl( c_3 r(u_1) r(u_0) + c_4 r(u_0) - c_1 + c_2 r(u_0)\bigr)} - \frac{r''(u_0)r(u_0) - r'^2(u_0)}{r(u_0)r'(u_0)},\label{alpha2}
\end{gather}
where $P(u_0)$ and $r(u_0)$ are arbitrary smooth functions, $c_1 \neq 0$, $c_3 \neq 0$, $c_2$, and $c_4$ are arbitrary constants.
\end{Theorem}

\begin{proof}By taking $k=2$ in the statement of Lemma~\ref{lemma4} we get
\begin{gather}
[D_x, Z_6] = -(\alpha_0 u_{0,x}+2 \alpha_1 u_{1,x})Z_6 + \cdots,\label{eq4_16} \\
[ D_x, Z_7] = -(\alpha_2 u_{2,x}+2 \alpha_1 u_{1,x})Z_7 - (Y_1(\alpha_1 u_{1,x})+2Y_1(\alpha_2 u_{2,x}))Z_4 +\cdots,\label{eq4_17} \\
[ D_x, Z_8 ] = - (\alpha_0 u_{0,x} + 2 \alpha_1 u_{1,x} + \alpha_2 u_{2,x})Z_8 + Y_0(\alpha_1 u_{1,x})Z_7 + Y_2(\alpha_1 u_{1,x})Z_6 \nonumber\\
\hphantom{[ D_x, Z_8 ] =}{} - \big( Y_1 (\alpha_1 u_{1,x})+ Y_1 (\alpha_0 u_{0,x} + \alpha_2 u_{2,x})\big) Z_5 + \cdots.\label{eq4_18}
\end{gather}
Now we apply the operator ${\rm ad}_{D_x}$ to both sides of \eqref{eq4_13} and then simplify due to the relations \eqref{eq4_16}, \eqref{eq4_17}, \eqref{eq4_18}. Comparison of the coef\/f\/icients before~$Z_7$ and~$Z_6$ implies $\lambda = 0$ and $\mu =0$. Thus formula \eqref{eq4_13} is simplif\/ied
\begin{gather} \label{eq4_19}
Z_8 = \nu Z_5 + \rho Z_4 + \kappa Z_3 + \sigma Z_2 + \delta Z_1 + \eta Z_0.
\end{gather}
In what follows we will use the following commutativity relations
\begin{gather}
[ D_x, Z_8 ] = -(a_2 + 2a_1 + a_0)Z_8 + Y_0(a_1) Z_7- Y_2(a_1)Z_6
- Y_1(a_2 + a_1 + a_0) Z_5 \nonumber \\
\hphantom{[ D_x, Z_8 ] = }{} + Y_1 Y_0 (a_1) Z_4 - Y_1 Y_2 (a_1) Z_3 + (Y_1 Y_2 Y_0 (a_1) +Z_5(a_1))Z_1,\label{eq4_20} \\
 [ D_x, Z_5 ] = -(a_0 + a_1 + a_2) Z_5 + Y_0(a_1)Z_4 - Y_2(a_1)Z_3 + Y_2 Y_0(a_1) Z_1. \label{eq4_21}
\end{gather}
Let us apply ${\rm ad}_{D_x}$ to \eqref{eq4_19} then simplify by using \eqref{eq4_20}, \eqref{eq4_21}, \eqref{eq4_19} and gather the coef\/f\/i\-cients at~$Z_5$
 \begin{gather*}
 -(a_2 + 2 a_1 + a_0) \nu - Y_1(a_2+a_1 + a_0) = D_x(\nu) - (a_2 + a_1 + a_0)\nu
 \end{gather*}
or the same
\begin{gather} \label{I45}
 D_x(\nu) = -a_1 \nu - Y_1(a_2 + a_1 + a_0).
 \end{gather}
Equation \eqref{I45} implies that $\nu$ depends on three variables $\nu = \nu(u,u_1,u_2)$ and splits down into three equations as follows
 \begin{gather}
 \nu_u = - \alpha_{0,u_1}, \label{I46} \\
 \nu_{u_1} = -\alpha_1 \nu - \alpha_{1, u_1} - \alpha^2_1, \label{I47} \\
 \nu_{u_2} = - \alpha_{2,u_1}. \label{I48}
 \end{gather}
 Substituting $\alpha$ def\/ined by \eqref{eq4_10} into \eqref{I46} and integrating with respect to $u$, we obtain
 \begin{gather} \label{I49}
 \nu = -\frac{P'(u_1)}{P(u_1)+ Q(u_0)} + H(u_1,u_2).
 \end{gather}
From equation \eqref{I48} we f\/ind
\begin{gather} \label{I50}
\nu = -\frac{Q'(u_1)}{P(u_2) + Q(u_1)} + R(u_0,u_1).
\end{gather}
Comparison of \eqref{I49} and \eqref{I50} yields
\begin{gather*}
-\frac{P'(u_1)}{P(u_1)+ Q(u_0)} + H(u_1,u_2) = -\frac{Q'(u_1)}{P(u_2) + Q(u_1)} + R(u,u_1).
\end{gather*}
Due to the fact that variables $u_0$, $u_1$, $u_2$ are independent we obtain
\begin{gather*}
-\frac{P'(u_1)}{P(u_1)+ Q(u_0)} -R(u_0,u_1) = -\frac{Q'(u_1)}{P(u_2) + Q(u_1)} - H(u_1,u_2) = -A(u_1).
\end{gather*}
Hence
\begin{gather*}
H(u_1, u_2) = -\frac{Q'(u_1)}{P(u_2)+Q(u_1)} + A(u_1)
\end{gather*}
and then
\begin{gather} \label{I51}
\nu = -\frac{P'(u_1)}{P(u_1)+Q(u_0)} - \frac{Q'(u_1)}{P(u_2)+Q(u_1)} + A(u_1).
\end{gather}
Note that $\lambda_{(R)}$ def\/ined by \eqref{eq4_9} satisf\/ies the equation \eqref{eq4_5_0}, i.e.,
\begin{gather*}
\lambda_{(R),u}= - \alpha_0 \lambda_{(R)} - \alpha_{0,u_0} -\alpha_0^2.
\end{gather*}
Then
\begin{gather} \label{I52}
\lambda_{(R)1,u_1} = - \alpha_1 \lambda_{(R)1} - \alpha_{1,u_1} -\alpha^2_1,
\end{gather}
where $\lambda_{(R)1} = D_n(\lambda_{(R)})$. Here $D_n$ is a shift operator $D_n u_k = u_{k+1}$. Let us subtract~\eqref{I52} from~\eqref{I47}
\begin{gather*}
\big(\nu - \lambda_{(R)1} \big)_{u_1} = -\alpha_1 \big(\nu - \lambda_{(R)1} \big).
\end{gather*}
Substituting functions \eqref{eq4_9} and \eqref{I51} into the last equation we arrive at the equality
\begin{gather*}
 -\frac{P'(u_1)B(u_1)}{P(u_1)+Q(u_0)} - \frac{Q'(u_1)B(u_1)}{P(u_2)+Q(u_1)} \nonumber\\
\qquad\quad{} +\frac{1} {2} \left( \log Q'(u_1)P'(u_1) \right)' \left( \frac{Q'(u_1)}{P(u_2)+Q(u_1)} - \frac{P'(u_1)}{P(u_1)+Q(u_0)} + B(u_1) \right) \nonumber\\
\qquad{}= \frac{Q''(u_1)}{P(u_2)+Q(u_1)} - \frac{P''(u_1)}{P(u_1)+Q(u_0)} + B'(u_1), 
\end{gather*}
where $B(u_1) = A(u_1) - \frac{Q''(u_1)}{Q'(u_1)}$. This equality is satisf\/ied only if the following conditions hold
\begin{gather}
Q''(u_1) = -Q'(u_1)B(u_1) + \frac{1} {2} Q'(u_1) \big( \log Q'(u_1) P'(u_1) \big)', \label{I54} \\
P''(u_1) = P'(u_1) B(u_1) + \frac{1} {2} P'(u_1) \big( \log Q'(u_1)P'(u_1) \big)', \label{I55} \\
B'(u_1) = \frac{1} {2} B(u_1) \big( \log Q'(u_1)P'(u_1) \big)'. \label{I56}
\end{gather}
The equation \eqref{I56} is satisf\/ied if $ B(u_1) = 0$ or
\begin{gather} \label{I57}
( \log B(u_1))' = \frac{1} {2} ( \log Q'(u_1) P'(u_1))'.
\end{gather}

If $B(u_1) = 0$ then $Q(u_1) = c_1 P(u_1) + c_2$ and{\samepage
\begin{gather}
 \alpha_0 = \frac{P'(u_0)}{P(u_0) + c_1 P(u_{-1}) + c_2} + \frac{c_1 P'(u_0)}{P(u_1) + c_1 P(u_0) + c_2} - \frac{P''(u_0)}{P'(u_0)}, \nonumber \\ 
 \lambda_{(M)}: =\nu = -\frac{P'(u_1)}{P(u_1) + c_1 P(u_0) + c_2} - \frac{c_1 P'(u_1)}{P(u_2) + c_1 P(u_1) + c_2} + \frac{Q''(u_1)}{Q'(u_1)},\label{lambda_M1} \\
 \lambda_{(R)} = -\frac{2 c_1 P'(u_0)}{P(u_1) + c_1 P(u_0) + c_2} + \frac{P''(u_0)}{P'(u_0)}, \label{lambda_R1} \\
 \lambda_{(L)} = -\frac{2 P'(u_1)}{P(u_1) + c_1 P(u_0) + c_2} + \frac{P''(u_1)}{P'(u_1)}. \label{lambda_L1}
\end{gather}
Here $c_1 \neq 0$.}

If $B(u_1) \neq 0$ then from the system of equations \eqref{I54}, \eqref{I55}, and \eqref{I57} we obtain that $Q(u_1) = -\frac{c_1} {r(u_1)} + c_2$, $P(u_1) = c_3 r(u_1) + c_4$ and
\begin{gather}
 \alpha_0 = \frac{c_3 r(u_{-1}) r'(u_0)}{c_3 r(u_0) r(u_{-1}) + c_4 r(u_{-1}) - c_1 + c_2 r(u_{-1} )} \nonumber\\
\hphantom{\alpha_0 =}{} + \frac{c_1 r'(u_0)}{r(u_0) \bigl( c_3 r(u_1) r(u_0) + c_4 r(u_0) - c_1 + c_2 r(u_0)\bigr)} - \frac{r''(u_0)r(u_0) - r'^2(u_0)}{r(u_0)r'(u_0)},\nonumber\\ 
 \lambda_{(M)}:= \nu = -\frac{c_3 r(u_0) r'(u_1)}{c_3 r(u_1) r(u_0) + c_4 r(u_0) - c_1 + c_2 r(u_0)}\nonumber \\
\hphantom{\lambda_{(M)}:= \nu =}{} -\frac{c_1 r'(u_1)}{r(u_1) \bigl( c_3 r(u_2) r(u_1) + c_4 r(u_1) - c_1 + c_2 r(u_1) \bigr)}\nonumber\\
\hphantom{\lambda_{(M)}:= \nu =}{}
+ \frac{r'(u_1)}{r(u_1)} - \frac{2 r'^2(u_1) - r''(u_1) r(u_1)}{r(u_1) r'(u_1)}, \label{lambda_M2}
\\
 \lambda_{(R)} = -\frac{2 c_1 r'(u_0)}{r(u_0)\bigl( c_3 r(u_1) r(u_0) + c_4 r(u_0) - c_1 + c_2 r(u_0) \bigr)}
+\frac{r''(u_0) r(u_0) -2 r'^2(u_0)}{r(u_0) r'(u_0)}, \label{lambda_R2} \\
 \lambda_{(L)} = \frac{-2c_3r(u_0)r'(u_1)}{\bigl( c_3 r(u_1) r(u_0) + c_4 r(u_0) - c_1 + c_2 r(u_0) \bigr)} + \frac{r''(u_1)}{r'(u_1)}. \label{lambda_L2}
\end{gather}

Now let us apply ${\rm ad}_{D_x}$ to \eqref{eq4_19} using \eqref{eq4_20}, \eqref{eq4_21}, \eqref{eq2_8} and the facts that $Z_4=D_n(Z_3)$ and $[Z_1, Z_4 ]=D_n[Z_0, Z_3] =-D_n(W_2)= -D_n(\lambda_{(R)})W_1 = D_n(\lambda_{(R)})Z_4$ and write down coef\/f\/icients before $Z_4$
\begin{gather*}
 -(a_2 + 2a_1 + a_0) \rho + Y_1 Y_0 (a_1) + Y_0 (a_1) D_n(\lambda_{(R)})= \nu Y_0(a_1) + D_x(\rho) -(a_1 + a_2)\rho.
\end{gather*}
Then
\begin{gather} \label{I62}
D_x(\rho) = -(a_1 + a_0) \rho + Y_1 Y_0(a_1) + Y_0(a_1) D_n(\lambda_{(R)})- \nu Y_0(a_1).
\end{gather}
 The equation \eqref{I62} implies that $\rho = \rho(u,u_1,u_2)$ and splits down into three equations as follows
 \begin{gather*}
 \rho_{u_2} = 0, \qquad -\alpha_0 \rho = \rho_{u_0}, \qquad
 - \alpha_1 \rho + \alpha_{1,u_0 u_1} + \alpha_1 \alpha_{1,u_0} + \alpha_{1,u_0} D_n(\lambda_{(R)})- \nu \alpha_{1,u_0} = \rho_{u_1}.
 \end{gather*}
If $\alpha_0$, $\nu$, and $\lambda_{(R)}$ are def\/ined by the formulas \eqref{alpha1}, \eqref{lambda_M1}, and \eqref{lambda_R1} or by the formulas \eqref{alpha2}, \eqref{lambda_M2}, and \eqref{lambda_R2} correspondingly then $\rho = 0$ and the last equations are satisf\/ied.

Now let us apply ${\rm ad}_{D_x}$ to \eqref{eq4_19} using \eqref{eq4_20}, \eqref{eq4_21}, \eqref{eq2_8}, \eqref{H-I_Y1Y10} and write down coef\/f\/icients before $Z_3$
\begin{gather*}
 -(a_2 + 2a_1 + a_0 ) \kappa - Y_1 Y_2(a_1) - Y_2(a_1)\lambda_{(L)} = -\nu Y_2(a_1) + D_x(\kappa) -(a_1 + a_2) \kappa.
\end{gather*}
Then
\begin{gather} \label{I67}
D_x(\kappa) = -(a_2 + a_1)\kappa - Y_1 Y_2(a_1) - Y_2(a_1) \lambda_{(L)} + \nu Y_2(a_1).
\end{gather}
The equations \eqref{I67} implies that $\kappa=\kappa(u_1,u_2)$ and splits down into two equations as follows
\begin{gather*}
\kappa_{u_2} = -\alpha_2 \kappa, \qquad
\kappa_{u_1} = -\alpha_1 \kappa - \alpha_{1,u_2 u_1} - \alpha_1 \alpha_{1, u_2} - \alpha_{1,u_2} \lambda_{(L)} + \nu \alpha_{1,u_2}.
\end{gather*}
If $\alpha_0$, $\nu$, and $\lambda_{(L)}$ are def\/ined by the formulas \eqref{alpha1}, \eqref{lambda_M1}, and \eqref{lambda_L1} or by the formulas \eqref{alpha2}, \eqref{lambda_M2}, and \eqref{lambda_L2} correspondingly then $\kappa = 0$ and the last equations are satisf\/ied.

Apply ${\rm ad}_{D_x}$ to \eqref{eq4_19} taking into account that $\rho = \kappa = 0$ and write down coef\/f\/icients before operators~$Z_2$,~$Z_1$ and~$Z_0$
\begin{gather*}
D_x(\sigma) = -(2a_1+ a_0)\sigma,\\
D_x(\delta) = -(a_2 + a_1 +a_0)\delta - Y_1 Y_2 Y_0(a_1) + \lambda Y_2 Y_0 (a_1),\\
D_x(\eta) = -(a_2 + 2a_1).
\end{gather*}
From these equations we obtain that $\sigma=\delta=\eta = 0$.

Thus we have proved that if the expansions \eqref{eq4_2}, \eqref{eq4_13} hold then \eqref{eq4_13} should be as follows
\begin{gather*}
Z_8 = \lambda_{(M)} Z_5.
\end{gather*}
Or the same
\begin{gather} \label{I_Y1Y210}
[ Y_1, [Y_2,[Y_1,Y_0]] ] = \lambda_{(M)} [Y_2,[Y_1,Y_0]],
\end{gather}
where $\lambda_{M}$ def\/ined by the formula \eqref{lambda_M1} or \eqref{lambda_M2} and $\alpha_0$, $\lambda_{(R)}$, and $\lambda_{(L)}$ are def\/ined by the~for\-mulas \eqref{alpha1}, \eqref{lambda_R1}, and \eqref{lambda_L1} or by the~for\-mulas \eqref{alpha2}, \eqref{lambda_R2}, and \eqref{lambda_L2} correspondingly.
\end{proof}

Corollary of Theorems \ref{theorem3} and \ref{theorem4}:
\begin{Corollary} \label{corollary2} In both cases $Q(u_1) = -\frac{c_1} {r(u_1)} + c_2$, $P(u_1) = c_3 r(u_1) + c_4$ and $Q(u_1) = c_1 P(u_1) + c_2$ the constraint~\eqref{Theor-I21} is satisfied identically.
 \end{Corollary}

In a similar way we check that the same conditions \eqref{alpha1}, \eqref{alpha2} provides the representations
\begin{gather*}
[Y_0, [Y_2,[Y_1,Y_0]] t] = \lambda_{(R)} [Y_2,[Y_1,Y_0]],\\
[Y_2,[Y_2,[Y_1,Y_0]] ] = D_n(\lambda_{(L)}[Y_2,[Y_1,Y_0]].
\end{gather*}

\section{Comments on the classif\/ication result}\label{section5}

In this section we brief\/ly discuss the statements of Theorems~\ref{theorem4} and~\ref{theorem5} (see below) claiming that the lattice \eqref{eq1} is integrable in the sense of Def\/inition~\ref{definition1} only for two choices of the function $\alpha$ given by \eqref{alpha1} and \eqref{alpha2}. In both cases the lattice has a functional freedom which is removed by an appropriate point transformation. Therefore we have

\begin{Theorem} \label{theorem6}
 Any lattice \eqref{eq1} integrable in the sense above is reduced by the point transformation $v=p(u)$ to the following lattice
\begin{gather}
 v_{n,xy}=v_{n,x}v_{n,y}\left(\frac{1} {v_n-v_{n-1} }-\frac{1} {v_{n+1} -v_{n}}\right). \label{alphalast1}
 \end{gather}
\end{Theorem}

Specify the point transformations\footnote{We are glad to acknowledge that these transformations are found by R.I.~Yamilov and R.N.~Garifullin (private communication).} applied to the lattices. Change of the variables $w=P(u)$ reduces \eqref{alpha1} to
\begin{gather} \label{eq_w}
w_{n,xy} = w_{n,x} w_{n,y} \left( \frac{1} {w_n+c_1 w_{n-1} + c_2} + \frac{c_1} {w_{n+1} + c_1 w_n + c_2} \right).
\end{gather}
The latter is connected with \eqref{alphalast1} by the change of the variables $v_n = (-c_1)^nw_n - \frac{c_2} {1+c_1} $ if $c_1 \neq -1$ and by $v_n = w_n - c_2 n $ in the special case $c_1 = -1$.

Change of the variables $v = r(u)$ reduces \eqref{alpha2} to
\begin{gather} \label{temp_eq1}
v_{n,xy} = v_{n,x} v_{n,y} \left( \frac{v_{n-1}}{v_n v_{n-1} + \beta v_{n-1} - \gamma} + \frac{v_{n+1} + \beta}{v_n v_{n+1} + \beta v_n - \gamma} \right),
\end{gather}
where we denote $\beta = \frac{c_2 + c_4}{c_3}$, $\gamma = c_1 / c_3$. Then change of the variables $v = \beta \big( \frac{1}{w} + c\big)$, $\gamma = \beta^2 (c^2 + c)$ reduces \eqref{temp_eq1} to
\begin{gather*}
w_{n,xy} = \frac{1}{w_n + \frac{c}{c+1} w_{n-1}+\frac{1}{c+1}}+\frac{\frac{c}{c+1}}{w_{n+1} + \frac{c}{c+1} w_n + \frac{1}{c+1}}.
\end{gather*}
The latter coincides with \eqref{eq_w} if $c_1 = \frac{c}{c+1}$, $c_2 = \frac{1}{c+1}$.

Note that equation \eqref{alphalast1} coincides with the Ferapontov--Shabat--Yamilov equation found in~\cite{ShY} and~\cite{Fer-TMF}.

\begin{Theorem} \label{theorem5}
The characteristic Lie rings in $x$- and $y$-directions for the following system of hyperbolic type equations
\begin{gather}
v_{-1} =c_0, \nonumber\\
v_{n,xy}=v_{n,x}v_{n,y}\left(\frac{1} {v_n-v_{n-1} }-\frac{1} {v_{n+1} -v_{n}}\right), \qquad 0\leq n\leq N,\label{eq31} \\
v_{N+1} =c_1, \nonumber
\end{gather}
are of finite dimension.
\end{Theorem}

The proof of Theorem~\ref{theorem5} can be found in Appendix~\ref{appendixA}.

\begin{Corollary}The system \eqref{eq31} is Darboux integrable.
\end{Corollary}

\begin{Remark}The following lattice (see \cite{ShY})
\begin{gather*}
q_{n,xy} = q_{n,x} q_{n,y} \bigl( f(q_{n+1}-q_n) - f(q_n - q_{n-1}) \bigr),\nonumber\\
f' = f^2 - b^2 
\end{gather*}
is reduced by the point transformation to \eqref{alphalast1}. Namely if $b \neq 0$ then $f(q) = b\tan (b(q+c)) = - ib \tanh(ib(q+c))$, where $c$ is the constant of integration, $i$ is the imaginary unit. So we have the lattice
\begin{gather*}
q_{n,xy} = q_{n,x} q_{n,y} (-ib)\bigl( \tanh (ib(q_{n+1}-q_n+c)) - \tanh (ib(q_n - q_{n-1}+c)) \bigr).
\end{gather*}
The change of variables $q_n = -\frac{i}{b} v_n - n c$ reduces the last lattice to \eqref{alphalast1}. If $b = 0$ then $f(q) = \frac{1}{q + c}$. By the change of variables $q_n = v_n - nc$ we obtain \eqref{alphalast1}.
\end{Remark}

\section{Conclusion}\label{section6}

In \cite{H2013} it was conjectured that any nonlinear integrable two-dimensional lattice of the form
\begin{gather}\label{end}
u_{n,xy}=g(u_{n+1},u_n,u_{n-1},u_{n,x},u_{n,y})
\end{gather}
admits cut-of\/f conditions reducing the lattice to a f\/inite system of the hyperbolic type partial dif\/ferential equations being integrable in the sense of Darboux when they are imposed at two points $n=N_1$ and $N_2$ chosen arbitrary.

In the present article we discussed the classif\/ication algorithm based on that conjecture. Actually we solved a problem of the complete description of the lattices \eqref{eq1} satisfying the suggested requirement. The lattice \eqref{eq1} is a particular case of the lattice \eqref{end} for which the mentioned cut-of\/f condition is easily found: $u_{N_1} =c_0$, $u_{N_2} =c_1$. This circumstance essentially simplif\/ies the situation. Nevertheless even in general when a priori the cut-of\/f condition is also unknown the algorithm might be ef\/fective since the assumption on the existence of such boundary conditions puts severe restrictions on the characteristic operators.

We show that the class of integrable lattices of the form \eqref{eq1} contains only one model up to the point transformations. This model coincides with the Ferapontov--Shabat--Yamilov equation. The one-dimensional reduction $x=y$ of this lattice satisf\/ies completely also the symmetry integrability conditions (see \cite{YamilovJPA06}).

\appendix

\section{Appendix}\label{appendixA}

The goal of the appendix is to prove Theorem~\ref{theorem5}. Let us introduce a special notation $Y_{i_k,\dots,i_0} $ for the multiple commutators. It is def\/ined consecutively
\begin{gather}\label{multi}
Y_{i_k,\dots,i_0} =[Y_{i_k},Y_{i_{k-1},\dots,i_0} ].
\end{gather}
Number $k$ is called the order of the operator \eqref{multi}.

In order to prove Theorem~\ref{theorem5} we show that the ring $R(y,N)$ is of f\/inite dimension. Actually we construct the basis in $R(y,N)$ containing the operators
\begin{gather} \label{I766}
\{Y_i\}_{i=0} ^{N},\quad \{Y_{i+1,i}\}_{i=0} ^{N-1}, \quad \{Y_{i+2,i+1,i}\}_{i=0} ^{N-2},\quad \ldots, \quad Y_{N,N-1,\ldots,0}.
\end{gather}

\subsection{The base case of the mathematical induction}\label{appendixA1}

In the previous section we have proved that
\begin{gather}
[Y_0, Y_{10} ] = \lambda_{(R)}Y_{10}, \qquad [Y_1, Y_{10} ] = \lambda_{(L)}Y_{10}, \label{eq5_2} \\
[Y_0, Y_{210} ] = \lambda_{(R)} Y_{210},\qquad [Y_1, Y_{210} ] = \lambda_{(M)} Y_{210}, \qquad [Y_2, Y_{210} ] = D_n(\lambda_{(L)}Y_{210}. \label{eq5_3}
\end{gather}
In what follows we will use the following relations which are easily verif\/ied
\begin{gather}
[D_x, Y_{3210} ]= -(a_3 + a_2 + a_1 + a_0) Y_{3210} -Y_3(a_2)Y_{210} + Y_0(a_1)Y_{321}, \label{I73} \\
[ D_x, [ Y_0, Y_{3210} ] ] = -(a_3 + a_2 + a_1 + 2a_0) [ Y_0, Y_{3210}] \nonumber\\
\hphantom{[ D_x, [ Y_0, Y_{3210} ] ] =}{} - Y_0(2a_1 + a_0) Y_{3210}
 - Y_3(a_2) [ Y_0, Y_{210} ] + Y_0 Y_0(a_1) Y_{321}, \label{I74} \\
 [ D_x, [Y_1, Y_{3210} ] ] = -(a_3 + a_2 + 2 a_1 + a_0) [ Y_1, Y_{3210} ] - Y_1(a_2 + a_1 +a_0)Y_{3210}\nonumber\\
 \hphantom{[ D_x, [Y_1, Y_{3210} ] ] =}{}
- Y_1 Y_3(a_2)Y_{210} - Y_3(a_2) [Y_1, Y_{210} ] + Y_1 Y_0(a_1)Y_{321} \nonumber\\
\hphantom{[ D_x, [Y_1, Y_{3210} ] ] =}{} + Y_0(a_1)[Y_1, Y_{321} ], \label{I75} \\
[D_x, [Y_2, Y_{3210} ]] = -(a_3 + 2 a_2 + a_1 + a_0)[Y_2, Y_{3210} ] -Y_2(a_3+a_2+a_1)Y_{3210}\nonumber\\
\hphantom{[D_x, [Y_2, Y_{3210} ]] =}{}
- Y_2 Y_3(a_2)Y_{210} -Y_3(a_2) [Y_2, Y_{210} ] + Y_2 Y_0(a_1)Y_{321} \nonumber\\
\hphantom{[D_x, [Y_2, Y_{3210} ]] =}{} + Y_0(a_1)[ Y_2, Y_{321} ], \label{I76} \\
[D_x, [ Y_3, Y_{3210} ] ] = -(2a_3 + a_2 + a_1 + a_0)[ Y_3, Y_{3210} ] -Y_3(a_3 + 2 a_2) Y_{3210} \nonumber\\
\hphantom{[D_x, [ Y_3, Y_{3210} ] ] =}{} - Y_3 Y_3(a_2) Y_{210} + Y_0(a_1) [ Y_3, Y_{321} ]. \label{I760}
\end{gather}

We prove the theorem by the mathematical induction. The base case consists in proving a~lot of the formulas concerned to small order commutators up to order six. When constructing a linear expression for a given element in $R(y,N)$ as a linear combination of those from \eqref{I766} we always use Lemma~\ref{lemma1}. That is why we need in explicit expressions for $[D_x,Y_{i_k,\dots,i_0} ]$. In the base case we prove a large set of the equalities. Since they all are proved by one and the same way we concentrate on one of them.

\begin{Lemma} \label{lemma6} We have
\begin{gather} \label{I_Y0Y3210}
[Y_0, Y_{3210}]=\lambda_{(R)} Y_{3210}.
\end{gather}
\end{Lemma}

\begin{proof} By applying the operator ${\rm ad}_{D_x}$ to $Z= [Y_0, Y_{3210}]-\lambda_{(R)} Y_{3210} $ and simplifying due to the equations \eqref{I73}, \eqref{I74} we obtain $[D_x,Z]=0$. Evidently $Z$ satisf\/ies the settings of Lemma~\ref{lemma1}. Due to this lemma we obtain $Z=0$. Lemma~\ref{lemma6} is proved. \end{proof}

In what follows we need in the formulas
\begin{gather}
[Y_1, Y_{3210} ]=\lambda_{(M)} Y_{3210},\label{I_Y1Y3210} \\
[Y_2, Y_{3210} ]=D_n\big(\lambda_{(M)}\big)Y_{3210}, \label{I_Y2Y3210} \\
[Y_3, Y_{3210} ]=D^2_n\big(\lambda_{(L)}\big)Y_{3210}, \label{I_Y3Y3210}
\end{gather}
which are some versions of the formula \eqref{I_Y0Y3210} from Lemma~\ref{lemma6}.

Now using formulas
\begin{gather}
 [ D_x, Y_{43210} ] = - (a_4+ a_3 + a_2 + a_1 + a_0) Y_{43210} - Y_4(a_3) Y_{3210} + Y_0(a_1) Y_{4321}, \label{I85} \\
 [D_x, [ Y_0, Y_{43210} ] ] = -(a_4+a_3+a_2+a_1+2a_0) [ Y_0, Y_{43210} ] \nonumber\\
\hphantom{[D_x, [ Y_0, Y_{43210} ] ] =}{} -Y_0(2a_1 + a_0)Y_{43210} - Y_4(a_3) [ Y_0, Y_{3210} ] + Y_0 Y_0 (a_1) Y_{4321},\label{I86} \\
 [ D_x, [ Y_1, Y_{43210} ] ] = -(a_4+a_3+a_2+2a_1 +a_0)[Y_1, Y_{43210} ] \nonumber\\
\hphantom{[ D_x, [ Y_1, Y_{43210} ] ] =}{} -Y_4(a_3)[ Y_1, Y_{3210} ] - Y_1(a_2+a_1 +a_0)Y_{43210} -Y_1 Y_4(a_3)Y_{3210} \nonumber \\
\hphantom{[ D_x, [ Y_1, Y_{43210} ] ] =}{}+Y_1 Y_0(a_1)Y_{4321} + Y_0(a_1) [ Y_1, Y_{4321} ], \label{I87} \\
 [ D_x, [ Y_2, Y_{43210} ] ] = -(a_4 + a_3 + a_2 + a_1 + a_0) [Y_2, Y_{43210} ] \nonumber\\
\hphantom{[ D_x, [ Y_2, Y_{43210} ] ] =}{} -Y_2 (a_3 + a_2 + a_1) Y_{43210} - Y_2 Y_4(a_3)Y_{3210} - Y_4(a_3)[ Y_2, Y_{3210} ] \nonumber\\
\hphantom{[ D_x, [ Y_2, Y_{43210} ] ] =}{} + Y_2 Y_0 (a_1) Y_{4321} + Y_0(a_1)[ Y_2, Y_{4321} ], \label{I88} \\
 [ D_x, [ Y_3, Y_{43210} ] ] = - (a_4 + 2a_3 + a_2 + a_1 + a_0) [ Y_3, Y_{43210} ] \nonumber\\
\hphantom{[ D_x, [ Y_3, Y_{43210} ] ] =}{} -Y_3(a_4 + a_3 + a_2) Y_{43210} - Y_3 Y_4(a_3) Y_{3210} - Y_4(a_3) [ Y_3, Y_{3210} ] \nonumber\\
\hphantom{[ D_x, [ Y_3, Y_{43210} ] ] =}{} +Y_3 Y_0(a_1)Y_{4321} + Y_0(a_1) [ Y_3, Y_{4321} ], \label{I89}\\
 [ D_x, [Y_4, Y_{43210} ] ] = -(2a_4 + a_3 + a_2 + a_1 + a_0) [ Y_4, Y_{43210} ] \nonumber\\
\hphantom{[ D_x, [Y_4, Y_{43210} ] ] =}{}
- Y_4(a_4 + 2a_3)Y_{43210} - Y_4 Y_4(a_3) Y_{3210} + Y_0(a_1)[ Y_4, Y_{4321} ] \label{I90}
\end{gather}
by direct calculations we prove that
\begin{gather}
[Y_0, Y_{43210} ] = \lambda_{(R)}Y_{43210}, \label{I92} \\
[ Y_1, Y_{43210} ] = \lambda_{(M)}Y_{43210}, \label{I93} \\
[ Y_2, Y_{43210} ] = D_n\big(\lambda_{(M)}\big)Y_{43210}, \label{I94} \\
[ Y_3, Y_{43210} ] = D^2_n\big(\lambda_{(M)}\big)Y_{43210}, \label{I95} \\
[ Y_4, Y_{43210} ] = D^3_n\big(\lambda_{(L)}\big)Y_{43210}, \label{I96}
\end{gather}
where $\lambda_{(R)}$, $\lambda_{(M)}$, and $\lambda_{(L)}$ are def\/ined by the formulas \eqref{lambda_R1}, \eqref{lambda_M1}, and \eqref{lambda_L1} or by the formulas \eqref{lambda_R2}, \eqref{lambda_M2}, and \eqref{lambda_L2} correspondingly.

Now, having explicit formulas for the small order commutators we are ready to work out an induction hypothesis.

\subsection{Inductive step}\label{appendixA2}

\begin{Theorem} \label{theorem7}
For $n>1$ the multi-commutators satisfy the following formulas
\begin{gather}
[ D_x, Y_{n+1,n,\ldots,0}] = -\left( \sum_{i=0} ^{n+1} a_i \right) Y_{n+1,n,\ldots,0} -Y_{n+1} (a_n)Y_{n,n-1,\ldots,0} + Y_0(a_1) Y_{n+1,n,\ldots,1}, \label{s2_eq1} \\
[ D_x, [ Y_0, Y_{n+1,n,\ldots,0} ] ] = -(a_{n+1} + \cdots + a_1 + 2 a_0)[ Y_0, Y_{n+1,n,\ldots,0} ] -Y_0(a_0+2a_1)Y_{n+1,n,\ldots,0}\nonumber\\
\hphantom{[ D_x, [ Y_0, Y_{n+1,n,\ldots,0} ] ] =}{} - Y_{n+1} (a_n) [ Y_0, Y_{n,n-1,\ldots,0} ] + Y_0 Y_0 (a_1) Y_{n+1,n,\ldots,1}, \label{s2_eq2} \\
[D_x, [ Y_k, Y_{n+1,n,\ldots,0} ] ] = -(a_{n+1} + \cdots + 2 a_k + \cdots + a_0) [ Y_k, Y_{n+1,n,\ldots,0} ] \nonumber\\
\hphantom{[D_x, [ Y_k, Y_{n+1,n,\ldots,0} ] ] =}{}
-Y_k\left( \sum_{i=0} ^{n+1} a_i \right)Y_{n+1,n,\ldots,0} - Y_k Y_{n+1} (a_n) Y_{n,n-1,\ldots,0} \nonumber\\
\hphantom{[D_x, [ Y_k, Y_{n+1,n,\ldots,0} ] ] =}{}
- Y_{n+1} (a_n)[ Y_k, Y_{n,n-1,\ldots,0} ] + Y_k Y_0 (a_1)Y_{n+1,n,\ldots,1}\nonumber\\
\hphantom{[D_x, [ Y_k, Y_{n+1,n,\ldots,0} ] ] =}{}
 +Y_0(a_1) [ Y_k, Y_{n+1,n,\ldots,1} ], \qquad k=1,2,\ldots,n,\label{s2_eq3} \\
[ D_x, [ Y_{n+1}, Y_{n+1,n,\ldots,0} ] ] = -(2a_{n+1} +a_n+\cdots+a_0) [ Y_{n+1}, Y_{n+1,n,\ldots,0} ] \nonumber\\
\hphantom{[ D_x, [ Y_{n+1}, Y_{n+1,n,\ldots,0} ] ] =}{}
-Y_{n+1} (a_{n+1} +2a_n)Y_{n+1,n,\ldots,0} - Y_{n+1} Y_{n+1} (a_n) Y_{n,n-1,\ldots,0} \nonumber\\
\hphantom{[ D_x, [ Y_{n+1}, Y_{n+1,n,\ldots,0} ] ] =}{}
+Y_0(a_1) [ Y_{n+1}, Y_{n+1,n,\ldots,1} ]. \label{s2_eq4}
\end{gather}
\end{Theorem}

\textbf{Proof by induction.} For $n=2$ and $n=3$ formulas \eqref{s2_eq1}--\eqref{s2_eq4} are previously proved (see \eqref{I73}--\eqref{I760} and \eqref{I85}--\eqref{I90}).

Assume that the multi-commutators satisfy the following formulas
\begin{gather*}
[ D_x, Y_{n,\ldots,0} ] = -\left( \sum_{i=0} ^{n} a_i \right) Y_{n,\ldots,0} -Y_{n}(a_{n-1} )Y_{n-1,\ldots,0} + Y_0(a_1) Y_{n,\ldots,1},\\
[ D_x, [ Y_0, Y_{n,\ldots,0} ] ] = -(a_{n}+ \cdots + a_1 + 2 a_0)[ Y_0, Y_{n,\ldots,0} ] \nonumber\\
\hphantom{[ D_x, [ Y_0, Y_{n,\ldots,0} ] ] =}{} -Y_0(a_0+2a_1)Y_{n,\ldots,0} - Y_{n}(a_{n-1}) [ Y_0, Y_{n-1,\ldots,0} ]
+ Y_0 Y_0 (a_1) Y_{n,\ldots,1}, \\ 
[D_x, [ Y_k, Y_{n,\ldots,0}] ] = -(a_{n}+ \cdots + 2 a_k + \cdots + a_0) [ Y_k, Y_{n,\ldots,0} ] -Y_k\left( \sum_{i=0} ^{n} a_i \right)Y_{n,\ldots,0} \nonumber\\
\hphantom{[D_x, [ Y_k, Y_{n,\ldots,0}] ] =}{}
- Y_k Y_{n}(a_{n-1}) Y_{n-1,\ldots,0} - Y_{n}(a_{n-1} )[ Y_k, Y_{n-1,\ldots,0} ] + Y_k Y_0 (a_1)Y_{n,\ldots,1}\\
\hphantom{[D_x, [ Y_k, Y_{n,\ldots,0}] ] =}{}
+Y_0(a_1) [ Y_k, Y_{n,\ldots,1} ], \qquad k=1,2,\ldots,n-1,\\ 
[ D_x, [ Y_{n}, Y_{n,\ldots,0} ] ] = -(2a_{n}+a_{n-1} +\cdots+a_0) [ Y_{n}, Y_{n,\ldots,0} ] \nonumber\\
\hphantom{[ D_x, [ Y_{n}, Y_{n,\ldots,0} ] ] =}{}
-Y_{n}(a_{n}+2a_{n-1} )Y_{n,\ldots,0} - Y_{n} Y_{n}(a_{n-1}) Y_{n-1,\ldots,0} +Y_0(a_1)[ Y_{n}, Y_{n,\ldots,1} ]. 
\end{gather*}
Then from these assumptions we deduce similar equations for $n+1$
\begin{gather}
[ D_x, Y_{n+1,n,\ldots,0} ] = [ D_x, [ Y_{n+1}, Y_{n,n-1,\ldots,0} ] ] \nonumber\\
\qquad {} = [Y_{n+1}, [D_x, Y_{n,n-1,\ldots,0} ] ] - [ Y_{n,n-1,\ldots,0}, [D_x, Y_{n+1} ] ] \nonumber\\
\qquad {} =[Y_{n+1}, -(a_n + a_{n-1} + \cdots + a_0)Y_{n,n-1,\ldots,0} - Y_n(a_{n-1} )Y_{n-1,n-2,\ldots,0} \nonumber\\
\qquad\quad{} + Y_0(a_1)Y_{n,n-1,\ldots,1} ] -[ Y_{n,n-1,\ldots,0}, -a_{n+1} Y_{n+1} ]\nonumber\\
\qquad {} = -\left( \sum_{i=0} ^{n+1} a_i \right) Y_{n+1,n,\ldots,0}
 - Y_{n+1} (a_n + a_{n-1} + \cdots + a_0)Y_{n,n-1,\ldots,0} Y_{n,n-1,\ldots,0} \nonumber\\
 \qquad\quad{} - Y_{n+1} Y_n(a_{n-1} )Y_{n-1,n-2,\ldots,0} - Y_n(a_{n-1} )[ Y_{n+1}, Y_{n-1,n-2,\ldots,0} ] \nonumber\\
 \qquad\quad {} + Y_{n+1} Y_0(a_1)Y_{n,n-1,\ldots,1} + Y_0(a_1) Y_{n+1,n,\ldots,1} +Y_{n,n-1,\ldots,0} (a_{n+1} )Y_{n+1}. \label{s2_eq9}
\end{gather}
Note that
\begin{gather} \label{s2_eq10}
Y_i(a_j) = 0 \qquad \mathrm{if} \quad |i-j|>1,
\end{gather}
i.e., if $i\neq j$, $i\neq j\pm 1$, and
\begin{gather*} 
[ Y_m, Y_{n, \ldots,0} ] = 0 \qquad \mathrm{if} \quad m-n>1
\end{gather*}
That is why the following terms in \eqref{s2_eq9} are equal to zero
\begin{gather*}
Y_{n+1} (a_{n-1} +\cdots+a_0) = 0, \qquad Y_{n+1} Y_n(a_{n-1} )=0, \\
[ Y_{n+1}, Y_{n-1,n-2,\ldots,0} ] = 0,\qquad Y_{n+1} Y_0(a_1) = 0, \qquad Y_{n,n-1,\ldots,0} (a_{n+1}) = 0.
\end{gather*}
Thus the equality \eqref{s2_eq9} takes the form \eqref{s2_eq1}.

Let us prove the formula \eqref{s2_eq2},
\begin{gather*}
[ D_x, [ Y_0, Y_{n+1,n,\ldots, 0} ] ]
= [Y_0, [ D_x, Y_{n+1,n,\ldots, 0} ] ] - [ Y_{n+1,n,\ldots, 0}, [ D_x, Y_0 ] ] \\
\hphantom{[ D_x, [ Y_0, Y_{n+1,n,\ldots, 0} ] ]}{}
=\left[ Y_0, -\!\left( \sum_{i=0} ^{n+1} a_i \!\right) \!Y_{n+1,n,\ldots,0} -Y_{n+1} (a_n)Y_{n,n-1,\ldots,0} + Y_0(a_1) Y_{n+1,n,\ldots,1} \right]\\
\hphantom{[ D_x, [ Y_0, Y_{n+1,n,\ldots, 0} ] ]=}{} -[ Y_{n+1,n,\ldots,0}, -a_0 Y_0].
\end{gather*}
From this equality using property of linearity of the commutators and the equations \eqref{s2_eq10} we obtain the formula \eqref{s2_eq2}. The formulas \eqref{s2_eq3} and \eqref{s2_eq4} are proved in a similar way.

\begin{Theorem} \label{theorem8}
For $m \geq 1$ the multi-commutators satisfy the following formulas
\begin{gather}
[ Y_0, Y_{m+1,m,\ldots,0} ] = \lambda_{(R)} Y_{m+1,m,\ldots,0}, \label{s2_eq12} \\
[ Y_k, Y_{m+1,m,\ldots,0} ] = D^{k-1} _n\big(\lambda_{(M)}\big)Y_{m+1,m,\ldots,0}, \qquad k=1,\ldots,m, \label{s2_eq13} \\
[ Y_{m+1}, Y_{m+1,m,\ldots,0} ] = D^{m}_n\big(\lambda_{(L)}\big)Y_{m+1,m,\ldots,0}. \label{s2_eq14}
\end{gather}
\end{Theorem}

\textbf{Proof by induction.} For $m=1,2,3$ formulas \eqref{s2_eq12}--\eqref{s2_eq14} are true
(see \eqref{eq5_2}, \eqref{eq5_3}, \eqref{I_Y1Y210}, \eqref{I_Y0Y3210}, \eqref{I_Y1Y3210}, \eqref{I_Y2Y3210}, \eqref{I_Y3Y3210}, \eqref{I92}--\eqref{I96}).

Assume that the multi-commutators satisfy the following formulas
\begin{gather}
[ Y_0, Y_{m,m-1,\ldots,0} ] = \lambda_{(R)} Y_{m,m-1,\ldots,0}, \label{s2_eq15} \\
[ Y_k, Y_{m,m-1,\ldots,0} ] = D^{k-1}_n\big(\lambda_{(M)}\big)Y_{m,m-1,\ldots,0}, \qquad k=1,\ldots,m-1, \label{s2_eq16} \\
[ Y_{m}, Y_{m,m-1,\ldots,0} ] = D^{m-1}_n\big(\lambda_{(L)}\big)Y_{m,m-1,\ldots,0}. \label{s2_eq17}
\end{gather}

Let us f\/irst prove the formula \eqref{s2_eq12}. The proof is rather tricky: we assume the expansion with undetermined coef\/f\/icients
\begin{gather}
[ Y_0, Y_{m+1,m,\ldots,0} ] = \lambda Y_{m+1,m,\ldots,0} + \mu Y_{m+1,m,\ldots,1} + \nu Y_{m,m-1,\ldots,0} +\epsilon Y_{m+1,m,\ldots,2} \nonumber \\
\hphantom{[ Y_0, Y_{m+1,m,\ldots,0} ] =}{} + \eta Y_{m,m-1,\ldots,1} +\zeta Y_{m-1,m-2,\ldots,0} + \cdots\nonumber\\
\hphantom{[ Y_0, Y_{m+1,m,\ldots,0} ] =}{}+ \theta Y_{m+1,m} + \cdots + \xi Y_{10} + \sigma Y_{m+1} + \cdots + \delta Y_0, \label{s2_eq18}
\end{gather}
and then evaluate the coef\/f\/icients consecutively in the following way. We apply the opera\-tor~${\rm ad}_{D_x}$ to \eqref{s2_eq18} and gather the coef\/f\/icients before the linearly independent operators. For instance, by comparing the coef\/f\/icients before the multi-commutator $Y_{m+1,m,\ldots,0} $ and then using formulas from Theorem~\ref{theorem1} with $n=m-1$ we f\/ind
\begin{gather*}
D_x(\lambda) = -a_0 \lambda - Y_0(a_0 + 2a_1 ).
\end{gather*}
The latter coincides with the equation \eqref{eq4_3} and, therefore, we can conclude that $\lambda = \lambda_{(R)}$.

Compare now the coef\/f\/icients before $Y_{m+1,m,\ldots,1} $ to get an equation for determining $\mu$.
Note that by Theorem~\ref{theorem7} we have
\begin{gather}
[ D_x, Y_{m+1,m,\ldots,1} ] = D_n [ D_x, Y_{m,m-1,\ldots,0} ] = |\text{Theorem~\ref{theorem7}} | \nonumber\\
\qquad{} = D_n \left( -\left(\sum_{i=0} ^m a_i \right)Y_{m,m-1,\ldots,0} -Y_m(a_{m-1} )Y_{m-1,m-2,\ldots,0} +Y_0(a_1)Y_{m,m-1,\ldots,1} \right) \nonumber\\
\qquad{}
=-\left(\sum_{i=1} ^{m+1} a_i \right)Y_{m+1,m,\ldots,1} - Y_{m+1} (a_m)Y_{m,m-1,\ldots,1} + Y_1(a_2)Y_{m+1,m,\ldots,2}. \label{s2_eq19}
\end{gather}
Due to the relation \eqref{s2_eq19} the desired equation reduces to the form
\begin{gather*}
D_x(\mu) = -2a_0 \mu + Y_0 Y_0(a_1) - \lambda Y_0(a_1).
\end{gather*}
It is easily checked that the equation has the only solution $\mu = 0$. Continuing this way we can prove that all of the other coef\/f\/icients $\nu,\epsilon, \dots, \delta$ in \eqref{s2_eq18} vanish. Now for the operator $Z=[ Y_0, Y_{m+1,m,\ldots,0} ] - \lambda_{(R)} Y_{m+1,m,\ldots,0} $ we have $[D_x,Z]=0$. Due to Lemma~\ref{lemma1} it implies $Z=0$. That completes the proof of the formula \eqref{s2_eq12}.

Now we prove the formula \eqref{s2_eq13}. To this end we assume that the equation holds
\begin{gather}
[ Y_k, Y_{m+1,m,\ldots,0} ] = \lambda Y_{m+1,m,\ldots,0} + \mu Y_{m+1,m,\ldots,1} + \nu Y_{m,m-1,\ldots,0} +\epsilon Y_{m+1,m,\ldots,2} \nonumber \\
\hphantom{[ Y_k, Y_{m+1,m,\ldots,0} ] =}{} + \eta Y_{m,m-1,\ldots,1} +\zeta Y_{m-1,m-2,\ldots,0} + \cdots\nonumber\\
\hphantom{[ Y_k, Y_{m+1,m,\ldots,0} ] =}{} + \theta Y_{m+1,m} + \cdots + \xi Y_{10} + \sigma Y_{m+1} + \cdots + \delta Y_0 \label{s2_eq21}
\end{gather}
with the coef\/f\/icients to be determined.

Let us apply ${\rm ad}_{D_x}$ to \eqref{s2_eq21} and write down the coef\/f\/icients before the operator $Y_{m+1,m,\ldots,0}$
\begin{gather*}
D_x(\lambda) = -a_k \lambda - Y_k\left(\sum_{i=0} ^{m+1} a_i\right).
\end{gather*}
Apply $D^{-(k-1)}_n$ to the last equation
\begin{gather*}
D_x \big( D^{-(k-1)}_n(\lambda) \big) = -a_1 D^{-(k-1)}_n(\lambda) - Y_1(a_2+a_1+a_0).
\end{gather*}
This equation coincides with the equation \eqref{I45} then $D^{-(k-1)}_n(\lambda) = \lambda_{(M)}$ and \begin{gather*}
\lambda = D^{k-1} _n\big(\lambda_{(M)}\big).
\end{gather*}

Note that due to the formula \eqref{s2_eq16} we have
\begin{gather}
 [ Y_k, Y_{m+1,m,\ldots,1}] = D_n [ Y_{k-1}, Y_{m,m-1,\ldots,0} ] \nonumber\\
\hphantom{[ Y_k, Y_{m+1,m,\ldots,1}]}{} =D_n \big( D_n^{k-2} (\lambda_{(M)})Y_{m,m-1,\ldots,0} \big) = D_n^{k-1} \big(\lambda_{(M)}\big)Y_{m+1,m,\ldots,1}. \label{s2_eq22}
\end{gather}
Apply ${\rm ad}_{D_x}$ to \eqref{s2_eq21} using the formulas from Theorem~\ref{theorem7} and the formula \eqref{s2_eq22} and write down the coef\/f\/icients before the multi-commutator $Y_{m+1,m,\ldots,1}$
\begin{gather*}
 -(a_{m+1} + \cdots + 2 a_k + \cdots + a_0) \mu + Y_k Y_0 (a_1) + Y_0(a_1)D^{k-1} _n(\lambda_{(M)} \\
\qquad {} =\lambda Y_0(a_1) + D_x(\mu) - \left(\sum_{i=1} ^{m+1} \right) \mu.
\end{gather*}
Since $\lambda = D^{k-1} _n(\lambda_{(M)})$ the latter can be brought to the form
\begin{gather} \label{s2_eq23}
D_x(\mu) = -(a_0 + a_k) \mu + Y_k Y_0(a_1).
\end{gather}
Evaluate the action of the product of the operators
\begin{gather*}
 Y_kY_0(a_1) = \left( \frac{\partial}{\partial u_k} + \alpha_k u_{k,x} \frac{\partial}{\partial u_{k,x}} \right) ( \alpha_{1,u} u_{1,x} )
= \begin{cases}
\alpha_{1,uu_1} u_{1,x}+\alpha_1 \alpha_{1,u} u_{1,x}, & k=1,\\
\alpha_{1,uu_2} u_{1,x}, & k=2,\\
0, & k>2.
\end{cases}
\end{gather*}

Thus if $k=1$ then the equality \eqref{s2_eq23} takes the form
\begin{gather*}
D_x(\mu) = -(\alpha_0 u_x + \alpha_1 u_{1,x})\mu + (\alpha_{1,uu_1} +\alpha_1 \alpha_{1,u})u_{1,x}.
\end{gather*}
This equation implies that $\mu = \mu(u,u_1)$ and splits down into two equations as follows
\begin{gather*}
\mu_u = - \alpha_0 \mu, \qquad \mu_{u_1} = -\alpha_1 \mu + \alpha_{1,uu_1} +\alpha_1 \alpha_{1,u}.
\end{gather*}
Then we can prove that $\mu=0$.

If $k=2$ then the equality \eqref{s2_eq23} takes the form
\begin{gather*}
D_x(\mu) = -(\alpha_0 u_x + \alpha_2 u_{2,x})\mu + \alpha_{1,uu_2} u_{1,x}.
\end{gather*}
This equation implies that $\mu = \mu(u,u_1,u_2)$ and splits down into three equations as follows
\begin{gather*}
\mu_u = - \alpha_0 \mu, \qquad \mu_{u_1} = \alpha_{1,uu_2}, \qquad \mu_{u_2} = -\alpha_2 \mu.
\end{gather*}
And then again $\mu=0$.

If $k>2$ then the equality \eqref{s2_eq23} takes the form
\begin{gather*}
D_x(\mu) = -(\alpha_0 u_x + \alpha_k u_{k,x})\mu.
\end{gather*}
This equation implies that $\mu = \mu(u,u_k)$ and splits down into two equations as follows
\begin{gather*}
\mu_u = - \alpha_0 \mu, \qquad \mu_{u_k} = -\alpha_k \mu.
\end{gather*}
Then $\mu=0$.

In a similar way we can verify that all of the coef\/f\/icients in \eqref{s2_eq21} vanish except $\lambda$. Thus due to Lemma~\ref{lemma1} formula \eqref{s2_eq13} is correct.

Now we check the formula \eqref{s2_eq14}. First we assume that the following decomposition takes place
\begin{gather}
[ Y_{m+1}, Y_{m+1,m,\ldots,0} ] = \lambda Y_{m+1,m,\ldots,0} + \mu Y_{m+1,m,\ldots,1} + \nu Y_{m,m-1,\ldots,0} +\epsilon Y_{m+1,m,\ldots,2}\nonumber\\
\hphantom{[ Y_{m+1}, Y_{m+1,m,\ldots,0} ] =}{} + \eta Y_{m,m-1,\ldots,1} +\zeta Y_{m-1,m-2,\ldots,0} +\cdots + \theta Y_{m+1,m} + \cdots \nonumber\\
\hphantom{[ Y_{m+1}, Y_{m+1,m,\ldots,0} ] =}{} + \xi Y_{10} + \sigma Y_{m+1} + \cdots + \delta Y_0 \label{s2_eq25}
\end{gather}
with undef\/ined factors.

Let us apply ${\rm ad}_{D_x}$ to \eqref{s2_eq25} and write down the coef\/f\/icients before the multi-commutator $Y_{m+1,m,\ldots,0} $
\begin{gather*}
D_x(\lambda) = -a_{m+1} \lambda - Y_{m+1} (a_{m+1} +2a_m).
\end{gather*}
Apply $D^{-m}_n$ to this equation{\samepage
\begin{gather*}
D_x\big( D^{-m}_n(\lambda) \big) = -a_1 D^{-m}_n(\lambda) - Y_1(a_1 + 2a_0).
\end{gather*}
This equation coincides with \eqref{eq4_3} then $D^{-m}_n(\lambda) = \lambda_{(L)}$ and $\lambda = D^m_n(\lambda_{(L)})$.}

Let us apply ${\rm ad}_{D_x}$ to \eqref{s2_eq25} and write down the coef\/f\/icients before the multi-commutator $Y_{m+1,m,\ldots,1} $
\begin{gather*}
-(2a_{m+1} +a_m+\cdots+a_0)\mu + Y_0(a_1) D^m_n(\lambda_{(L)})= \lambda Y_0(a_1) + D_x(\mu) - \left( \sum_{i=1} ^{m+1} a_i \right) \mu.
\end{gather*}
Note that $\lambda = D^m_n(\lambda_{(L)})$ then the last equality takes the form
\begin{gather*}
D_x(\mu) = -(a_{m+1} + a_0) \mu.
\end{gather*}
Then $\mu = 0$.

Let us apply ${\rm ad}_{D_x}$ to \eqref{s2_eq25} and write down the coef\/f\/icients before the multi-commutator $Y_{m,m-1,\ldots,0}$
\begin{gather} \label{s2_eq26}
D_x(\nu) = -2 a_{m+1} \nu - Y_{m+1} Y_{m+1} (a_m) + \lambda Y_{m+1} (a_m).
\end{gather}
Note that
\begin{gather*}
 \lambda Y_{m+1} (a_m) - Y_{m+1} Y_{m+1} (a_m) \\
\qquad{} =D^m_n(\lambda_{(L)})Y_{m+1} (a_m) - Y_{m+1} Y_{m+1} (a_m) =D_m (\lambda_{(L)}) Y_1(a_0) - Y_1 Y_1(a_0)) = 0.
\end{gather*}
Then the equation \eqref{s2_eq26} takes the form $D_x(\nu) = -2 a_{m+1} \nu$ and we obtain that $\nu =0$.

In a similar way we can prove the vanishing of the other coef\/f\/icients in \eqref{s2_eq25}. Now by applying Lemma~\ref{lemma1} it is easy to complete the proof of the formula \eqref{s2_eq14}. Theorem~\ref{theorem8} is proved.

\subsection*{Acknowledgments}

The authors are grateful to the anonymous referees for their critical remarks and fruitful recommendations.

\pdfbookmark[1]{References}{ref}
\LastPageEnding

\end{document}